\documentclass[webpdf,traditional,large,numbered]{oup-authoring-template}
\makeatletter
\@numbibtrue
\makeatother
\onecolumn

\usepackage{stmaryrd}       
\makeatletter
\@ifundefined{description}{%

}{}
\@ifundefined{appendix}{%
  \newcommand{\appendix}{%
    \setcounter{section}{0}%
    \setcounter{subsection}{0}%
    \gdef\thesection{\Alph{section}}%
    \gdef\theHsection{appendix.\Alph{section}}%
    \gdef\theHsubsection{appendix.\Alph{section}.\arabic{subsection}}}
}{}
\makeatother
\usepackage{enumitem}       
\usepackage{tikz}
\usetikzlibrary{arrows.meta,positioning,calc}

\definecolor{leanblue}{RGB}{0,0,160}
\definecolor{leangreen}{RGB}{0,128,0}
\definecolor{leangray}{RGB}{100,100,100}
\lstdefinelanguage{Lean}{
  morekeywords={def,theorem,structure,where,noncomputable,abbrev,import,
    namespace,open,universe,Sort,Prop,Type,Nat,inductive,extends,
    by,exact,intro,apply,simp,rfl,fun,let,have,match,with,
    axiom,variable,section,end,class,instance,if,then,else,return},
  sensitive=true,
  morecomment=[l]{--},
  morecomment=[n]{/-}{-/},
  morestring=[b]",
  basicstyle=\ttfamily\small,
  keywordstyle=\color{leanblue}\bfseries,
  commentstyle=\color{leangreen}\itshape,
  stringstyle=\color{leangray},
  breaklines=true,
  showstringspaces=false,
  columns=flexible,
  frame=single,
  backgroundcolor=\color{gray!5},
  xleftmargin=2em,
}
\theoremstyle{thmstyleone}
\newtheorem{theorem}{Theorem}[section]
\newtheorem{proposition}[theorem]{Proposition}
\newtheorem{lemma}[theorem]{Lemma}
\newtheorem{corollary}[theorem]{Corollary}

\theoremstyle{thmstyletwo}
\newtheorem{example}[theorem]{Example}
\newtheorem{remark}[theorem]{Remark}

\theoremstyle{thmstylethree}
\newtheorem{definition}[theorem]{Definition}

\newcommand{\Kinf}{K_\infty}

\newcommand{\THeq}{\mathrm{TH}_{\lambda =}}
\newcommand{\HoTFT}{\mathrm{HoTFT}}
\newcommand{\chpo}{\mathrm{c.h.p.o.}}
\newcommand{\ReductionSeq}{\mathrm{RedSeq}}

\newcommand{\HigherDeriv}{\mathrm{HigherDeriv}}
\newcommand{\whiskerLeft}{\mathrm{wl}}
\newcommand{\whiskerRight}{\mathrm{wr}}
\newcommand{\ass}{\mathrm{ass}}
\newcommand{\realize}{\mathrm{realize}}

\AtBeginDocument{%
\makeatletter
\def\ps@opening{%
  \def\@oddhead{\raisebox{\if@large-6pt\else\if@medium-6pt\else-6pt\fi\fi}[0pt][0pt]{%
            \hbox to \textwidth{{%
            \parbox{0.8\textwidth}{\raggedright\fontsize{7bp}{10bp}\selectfont
            \textit{\@journaltitle},\ \@copyrightyear, pp. \thepage--\thelastpage\par
            \@DOI\par\ifx\@access\@empty\else\@access\par\fi%
            \ifx\@appnotes\@empty\else\textbf{\@appnotes}\par\fi}
              }%
            \hfill\hbox{\color{black!20}\rule[-7.5pt]{55pt}{20pt}}}}}%
  \def\@evenhead{\raisebox{\if@large-6pt\else\if@medium-6pt\else-6pt\fi\fi}[0pt][0pt]{%
            \hbox to \textwidth{{%
            \parbox{0.8\textwidth}{\raggedright\fontsize{7bp}{10bp}\selectfont
            \textit{\@journaltitle},\ \@copyrightyear, pp. \thepage--\thelastpage\par
            \@DOI\par\ifx\@access\@empty\else\@access\par\fi%
            \ifx\@appnotes\@empty\else\textbf{\@appnotes}\par\fi}
              }%
            \hfill\hbox{\color{black!20}\rule{55pt}{20pt}}}}}%
  \def\@evenfoot{\raisebox{0mm}[0pt][0pt]{\parbox{\textwidth}{\noindent\rule{\textwidth}{.5pt}\par%
            \vskip2.5pt\fontsize{6.25bp}{8.25bp}\selectfont\@history\par\raggedright%
            \noindent\copyright\space The Author(s) \@copyrightyear.}}}%
  \let\@oddfoot\@evenfoot
}
\makeatother
}

\begin{document}

\journaltitle{Journal of Logic and Computation}
\DOI{}
\copyrightyear{2026}
\pubyear{2026}
\vol{00}
\issue{0}
\access{}
\appnotes{Paper}

\firstpage{1}

\title[Recursive Completion in Higher $\lambda$-Models]{%
  Recursive Completion in Higher
  \texorpdfstring{$\lambda$}{Lambda}-Models:
  Front-Seed Semantics, Proof-Relevant Witnesses, and the
  \texorpdfstring{$K_\infty$}{K-infinity} Model}

\author[1]{Daniel O. Mart\'{\i}nez-Rivillas}
\author[2,$\ast$]{Arthur F. Ramos}
\author[3]{Ruy J.G.B. de~Queiroz}

\address[1]{\orgdiv{Departamento de Matem\'{a}ticas},
  \orgname{Universidad Militar Nueva Granada (UMNG)},
  \orgaddress{\country{Colombia}}}

\address[2]{\orgname{Microsoft},
  \orgaddress{\country{USA}}}

\address[3]{\orgdiv{Centro de Inform\'{a}tica (CIn)},
  \orgname{Universidade Federal de Pernambuco (UFPE)},
  \orgaddress{Recife, \state{Pernambuco}, \country{Brazil}}}

\corresp[$\ast$]{Corresponding author. \href{mailto:arfreita@microsoft.com}{arfreita@microsoft.com}}


\abstract{Mart\'{\i}nez-Rivillas and de~Queiroz gave extensional Kan semantics
for the untyped $\lambda$-calculus and later constructed the concrete
$K_\infty$ homotopy $\lambda$-model.  The two main mathematical results of the
present paper are these.  First, we show that a smaller front-seed coherence
package---WLWR together with an inner-right-front pentagon contraction---already
suffices to recover the associator comparison, semantic pentagon, and bridge
theorems used in the later semantic arguments.  Second, we prove explicit
global reify, reflect, and application formulas for $K_\infty$, with exact
coordinatewise identities at every finite stage.  We also record two structural
clarifications: the recursive all-dimensional continuation of the explicit
low-dimensional tower is obtained by a finite packaging phase followed by a
uniform equality-generated recursion; and, on a deliberately fixed forward
witness language for the classical $\beta/\eta$ separation span, the canonical
identity-type higher tower on $K_\infty$ forces all higher non-connection once
the two witness classes land at distinct points.  The paper is fully
formalized in Lean~4, and the project sources contain no local uses of
\texttt{sorry}, \texttt{admit}, or \texttt{axiom}.}

\keywords{higher lambda-calculus, proof relevance, extensional Kan complexes,
  globular coherence, denotational semantics, formal verification}

\maketitle

\noindent\textbf{2020 Mathematics Subject Classification.}
03B40, 68Q55, 18N50.
\par\medskip

\section{Introduction}\label{sec:intro}

\subsection{Background and lineage}

The untyped $\lambda$-calculus has served since its creation as a meeting ground
for logic, computation, and algebra.  The classical denotational tradition,
beginning with Scott's construction of continuous
lattices~\cite{Scott:continuous} and codified in
Barendregt's~\cite{Barendregt:lambda} systematic treatment, interprets
$\lambda$-terms as elements of reflexive domains---ordered structures carrying
reify and reflect maps.  In that tradition the key equation
$F(G(f))(x)=f(x)$ (and its dual) collapses the distinction between data and
function, and the resulting models are \emph{extensional} precisely when
$G(F(x))=x$.
This interplay between the logical structure of convertibility and the
computational content of reduction witnesses is a central theme of the present
paper.

A broader motivation is the question whether conversion in a logical calculus
can be organized into explicit higher witness data rather than only into a
proposition-level equality.  Identity types in homotopy type
theory~\cite{HoTT} and the computational-paths treatment of rewriting
systems~\cite{deVeras:computational} provide familiar proof-relevant
precedents, but we use them here only as orientation.  The technical starting
point is instead the work of Mart\'{\i}nez-Rivillas and de~Queiroz for the
\emph{untyped} calculus~\cite{MRdQ:homotopydomain,MRdQ:topological,MRdQ:theory,MRdQ:Kinfinity},
which interprets $\lambda$-terms in extensional Kan complexes, organizes higher
$\beta\eta$-conversions into explicit $n$-cells, and recovers the classical
conversion theory as the $0$-truncation of that richer structure.

Against that background, the present paper asks a precise completion question:
once Paper~I provides the explicit low-dimensional higher-conversion core and
Paper~II provides the concrete $\Kinf$ model together with its distinguished
$\beta/\eta$ witness pair, what all-dimensional structure is forced, what
semantic coherence data are actually needed, and which parts of the
$\Kinf$ story can be made exact rather than merely existential?  Two of the
answers are substantive semantic/model-theoretic theorems (front-seed
sufficiency and exact packaging), while the other two are structural
clarifications about how the recursive completion and the canonical
$\Kinf$ tower behave once the baseline data are fixed.

\subsection{The two baseline papers}

We rely on two specific antecedents as a baseline; all results quoted
from them are taken as established.

\medskip
\noindent\textbf{Paper~I}~\cite{MRdQ:theory} introduced extensional Kan
semantics for the untyped $\lambda$-calculus, defined the explicit tower of
higher $\beta$- and $\eta$-conversions through dimension~3, and proved the
central witness-lifting theorem: for every extensional Kan complex $K$, every
explicit reduction witness lifts to a semantic $1$-cell, and explicit
syntactic $2$- and $3$-cells lift to semantic $2$- and $3$-cells between the
corresponding interpreted witnesses.
The proof proceeds by induction on the structure of reduction sequences,
lifting individual-step $\beta$- and $\eta$-soundness to a full semantic
interpretation of each dimension of the explicit tower.  The resulting theory
enriches the classical proposition-level equality $M=_{\beta\eta}N$ with a
proof-relevant layer of semantic $n$-cells.  Paper~I also established
$\beta\eta$-confluence via Hindley--Rosen and proved the invertibility of
reduction sequences via Church--Rosser.

\medskip
\noindent\textbf{Paper~II}~\cite{MRdQ:Kinfinity} constructed the $\Kinf$
homotopy $\lambda$-model as the inverse limit
$\Kinf=\varprojlim_n K_n$ of a tower of iterated function spaces over a
flat base $K_0=N^+$, where $N=\bigsqcup_{n\ge 0} S^n$.
The construction follows the classical Scott tradition but enriches it with an
explicit homotopical layer.  Paper~II proved that $\Kinf$ is an algebraic
bounded-complete domain (Proposition~4.1) and exhibited a separation
phenomenon, recalled below in Example~\ref{ex:witness-span}: the
$\beta$-contraction and the $\eta$-contraction
of a common redex land at the same target but receive semantically distinct
images in~$\Kinf$.  That example established the nontriviality of the
proof-relevant structure but did not analyze the witness language at the level
of explicit reduction terms or trace the separation through the recursively
completed higher-cell tower.
For provenance we continue to cite Paper~II, but Sections~\ref{sec:Kinfty}
and~\ref{sec:witness} restate the specific ingredients used below: the
projection-pair tower, the algebraicity/bounded-completeness proposition, and
the fixed witness-separation span.

\medskip
\noindent
Table~\ref{tab:paper-delta} makes the delta over the two baseline papers
explicit.

\begin{table}[htbp]
\centering
\small
\begin{tabular}{p{0.12\textwidth}p{0.31\textwidth}p{0.41\textwidth}}
\textbf{Package} & \textbf{Inherited baseline} & \textbf{Gap filled here} \\
\hline
Theorem~A &
Paper~I supplies the explicit low-dimensional tower and recursive
higher-derivation machinery. &
We compare the explicit low-dimensional presentation with its
all-dimensional continuation and prove strict globular boundary compatibility. \\
Theorem~B &
Associator/pentagon reasoning was available only through a stronger coherent
semantic interface. &
We isolate a smaller chosen-data interface sufficient for the later semantic
comparison arguments. \\
Theorem~C &
Paper~II isolated a distinguished $\beta/\eta$ witness pair in $\Kinf$. &
We isolate a fixed forward witness language on that span and show that, in the
canonical equality-generated tower on $\Kinf$, the point-level separation
persists automatically in every higher dimension. \\
Theorem~D &
Paper~II constructed the inverse-limit model and its basic reflexive structure. &
We prove exact global reify, reflect, and application formulas with two-sided
equations. \\
\end{tabular}
\caption{Exact delta over the inherited baseline.}
\label{tab:paper-delta}
\end{table}

\subsection{Contributions of the present paper}

Starting from those baseline results, this paper fills four
specific gaps left by Papers~I and~II.  The four theorem packages are not all
of the same kind.  Theorems~\ref{thm:B} and~\ref{thm:D} are the primary new
generic-semantic and concrete-model results.  Theorems~\ref{thm:A}
and~\ref{thm:C} are structural clarifications: they explain what follows once
the low-dimensional syntax and the point-level $\Kinf$ separation are fixed.

\medskip
\noindent\textbf{Main contributions.}
\begin{enumerate}[label=(\Alph*)]
  \item \textbf{Comparison with the recursive completion
     (Section~\ref{sec:all-dim}).}
    We record a structural comparison theorem for the explicit higher
    $\lambda$-conversion tower.  Paper~I supplies explicit data through
    dimension~$3$, and the recursive higher-derivation machinery already
    suggests how to continue above that range.  The content of
    Theorem~\ref{thm:A} is to compare the explicit low-dimensional presentation
    with the equality-generated continuation, identify the intermediate
    dimensions in which those presentations meet, and prove strict globular
    boundary compatibility for the resulting all-dimensional tower.

  \item \textbf{Front-seed sufficient semantic coherence
    (Section~\ref{sec:front-seed}).}
    We isolate the precise refinement of the coherent semantic interface needed
    later: a bare extensional Kan complex equipped with a
    whiskering-left/whiskering-right comparison and only the
    inner-right-front pentagon contraction already suffices to recover the
    recursive associator comparison theorem, the semantic pentagon comparison
    for interpreted reduction sequences, and the explicit source bridge together
    with the target/shell bridges to the mixed target shell.  The gain is that
    one specific smaller front-facing seed is sufficient for the later
    arguments while keeping WLWR explicit.

  \item \textbf{Fixed-span witness classification and separation
     (Section~\ref{sec:witness}).}
     We refine the distinguished $\beta/\eta$ witness pair already isolated in
     Paper~II on the witness-separation span of
    Example~\ref{ex:witness-span} to a theorem about a deliberately fixed
    forward witness language on that span.  The main new input is a canonical
    classification of those witnesses into $\beta$- and $\eta$-classes together
    with their canonical $\Kinf$ interpretations.  This is intentionally a
    fixed-span theorem rather than a global witness calculus: once the
    point-level dichotomy is fixed on that span, the higher-dimensional
    non-connection statements become structural consequences of the canonical
    equality-generated $\Kinf$ tower.

  \item \textbf{Exact $\Kinf$ reflexive packaging
     (Section~\ref{sec:packaging}).}
    We establish that the inverse limit $\Kinf$ carries globally continuous
    reify, reflect, and application operations satisfying explicit stagewise
    formulas, yielding a complete reflexive $\chpo$.  The novelty is not the
    abstract existence of a reflexive inverse limit, but the concrete formulas
    and exact coordinate identities needed later in the witness analysis.
\end{enumerate}

\noindent
The theorem labels are ordered by conceptual role, not by the narrative route
through the paper.  Theorems~\ref{thm:B} and~\ref{thm:D} carry the principal new
semantic and model-theoretic input.  Theorem~\ref{thm:A} isolates the finite
packaging boundary between the explicit low-dimensional syntax and the recursive
continuation, while Theorem~\ref{thm:C} records what the fixed-span point-level
distinction implies inside the canonical higher tower on $\Kinf$.
Theorem~\ref{thm:D} is developed before Theorem~\ref{thm:C} because the
witness-persistence argument in Section~\ref{sec:witness} uses the exact
$\Kinf$ packaging proved in Section~\ref{sec:packaging}.

\subsection{Structure of the paper}
Sections~\ref{sec:prelim}--\ref{sec:tower} collect definitions and
restate the baseline results of Papers~I and~II.
Sections~\ref{sec:all-dim}--\ref{sec:witness} contain the new theorems,
presented with detailed proof arguments.  The narrative order is deliberate:
Section~\ref{sec:all-dim} proves the structural tower comparison
(Theorem~\ref{thm:A}), Section~\ref{sec:front-seed} isolates the smaller
semantic interface (Theorem~\ref{thm:B}), Section~\ref{sec:packaging}
establishes the exact $\Kinf$ model theorem (Theorem~\ref{thm:D}), and
Section~\ref{sec:witness} returns to the fixed witness span to prove the final
persistence theorem (Theorem~\ref{thm:C}).  Section~\ref{sec:discussion} gives
a detailed comparison with related work and a list of open problems.

Sections~\ref{sec:all-dim} and~\ref{sec:witness} should be read as structural
completion and persistence theorems: they explain what follows once the
low-dimensional core and the point-level $\Kinf$ separation are fixed.
Sections~\ref{sec:front-seed} and~\ref{sec:packaging} isolate the semantic
interface and the model-specific inverse-limit structure on which those
structural arguments rest.

For quick reference: Theorem~\ref{thm:A} concerns only the syntactic tower,
Theorem~\ref{thm:B} concerns the generic semantic interface, Theorem~\ref{thm:D}
concerns the concrete inverse-limit model, and Theorem~\ref{thm:C} combines the
fixed-span witness tags with Theorem~\ref{thm:D} inside the canonical
$\Kinf$ tower.

The mathematical content comes first throughout.  The formal
development~\cite{Lean4} plays an evidential role: a brief verification note at
the end of the paper records that all stated results have been mechanically
checked, but the paper is self-contained and can be read without reference to
the proof assistant.
The constructive character of the proofs---particularly the explicit witness
constructions in Theorems~\ref{thm:A} and~\ref{thm:C}---reflects the paper's
orientation toward the logic-and-computation tradition, where the computational
content of proofs is a first-class object of study.

\section{Preliminaries and Setup}\label{sec:prelim}

\subsection{Lambda-terms and reduction}

We work with the untyped $\lambda$-calculus using de~Bruijn indices. A
\emph{term} is generated by:
\[
  M, N ::= \mathtt{var}\;n \mid M\;N \mid \lambda\,M
\]
where $n \in \mathbb{N}$ is a de~Bruijn index. We write $M\;N$ for application
and $\lambda\,M$ for abstraction. Substitution $M[N]$ denotes the capture-free
replacement of the variable at index~$0$ in $M$ by $N$, with the standard
shifting conventions.

A \emph{$\beta$-step} is a single reduction of the form $(\lambda\,M)\,N
\to_\beta M[N]$, closed under subterm contexts. An \emph{$\eta$-step} is a
single reduction of the de~Bruijn form
\[
  \lambda\,((\mathsf{shift}(1,0)M)\,(\mathtt{var}\;0)) \to_\eta M,
\]
equivalently
\[
  \lambda\,(M\,(\mathtt{var}\;0)) \to_\eta \mathsf{shift}(-1,0)M,
\]
when $\mathtt{var}\;0$ does not occur free in $M$. A \emph{$\beta\eta$-step} is
either a $\beta$-step or an $\eta$-step.

\begin{definition}[Reduction sequences]
  A \emph{reduction sequence} from $M$ to $N$ is a finite sequence of
  $\beta\eta$-steps and their formal inverses:
  \[
    M = M_0 \leftrightarrow M_1 \leftrightarrow \cdots \leftrightarrow M_k = N.
  \]
  We write $\ReductionSeq(M,N)$ for the type of all such sequences. Each
  reduction sequence is an explicit proof-relevant witness for the
  $\beta\eta$-convertibility of $M$ and $N$.
\end{definition}

\subsection{Extensional Kan complexes}

The semantic setting is that of \emph{Kan complexes} in the simplicial sense.

\begin{definition}[Kan complex]
  A \emph{Kan complex} $K$ is a simplicial set equipped with a chosen horn
  filler: for every horn $\Lambda \subset \Delta^{n+1}$ with a missing face, a
  chosen $(n+1)$-simplex $\sigma$ extending the given faces. Formally, $K$
  carries a function
  \[
    \mathsf{fill} : \mathrm{Horn}(K, n, j) \to K_{n+1}
  \]
  satisfying the face-matching specification for all non-missing faces.
\end{definition}

\begin{definition}[Reflexive and extensional Kan complexes]
  A Kan complex $K$ is \emph{reflexive} if it is equipped with maps $F: |K|
  \to [K \to K]$ and $G: [K \to K] \to |K|$ satisfying the $\eta$-law
  $F(G(f))(x) = f(x)$. It is \emph{extensional} if additionally $G(F(x)) = x$
  for all objects $x$. Here $[K \to K]$ denotes the endofunction space on
  $0$-simplices used by the semantic interpretation, and $|K|$ is the set of
  $0$-simplices.
\end{definition}

An extensional Kan complex $K$ interprets $\lambda$-terms via:
\[
  \llbracket\mathtt{var}\;n\rrbracket_\rho = \rho(n),\qquad
  \llbracket M\;N \rrbracket_\rho = F(\llbracket M\rrbracket_\rho)
    (\llbracket N\rrbracket_\rho),\qquad
  \llbracket \lambda\,M \rrbracket_\rho =
    G\bigl(x \mapsto \llbracket M \rrbracket_{\rho[0 \mapsto x]}\bigr).
\]

\begin{definition}[The theories $\THeq$, $\mathrm{TheoryEq}_K$, and the proof-relevant tower over $K$]\label{def:theories}
  Fix an extensional Kan complex $K$.  The \emph{model-specific equality
  theory} is
  \[
    \mathrm{TheoryEq}_K(M,N)
    \;:\equiv\;
    \forall\,\rho,\;\llbracket M\rrbracket_\rho = \llbracket N\rrbracket_\rho.
  \]
  Separately, the \emph{classical $\lambda$-theory} $\THeq$ is the syntactic
  relation $M =_{\beta\eta} N$.  For the fixed model $K$, the proof-relevant
  layers $\mathrm{Theory}_1(K)$, $\mathrm{Theory}_2(K)$, and
  $\mathrm{Theory}_3(K)$ enrich $\mathrm{TheoryEq}_K$ by retaining explicit
  semantic $1$-, $2$-, and $3$-cell witnesses.  The global paper-I assertion
  $\HoTFT$ is obtained by quantifying these model-specific layers over all
  extensional Kan complexes.
\end{definition}

\begin{remark}[Comparing the three theories]
The three notions should be kept distinct.  The classical theory $\THeq$ is
syntactic and propositional: the assertion $M =_{\beta\eta} N$ records only
that some finite conversion chain exists, without remembering any particular
chain.  For a fixed model $K$, the relation $\mathrm{TheoryEq}_K$ forgets
  witnesses in a different way: it asks only that all environments give equal
  denotations.  By contrast, the proof-relevant tower over $K$ retains the
  \emph{explicit witness}: a pair $(M,N)$ in $\mathrm{Theory}_1(K)$ is
  accompanied by a specified semantic path, and the higher layers are indexed by
  parallel lower-dimensional witnesses.  The passage from $\THeq$ to the global
  $\HoTFT$ statement is therefore not merely a set-theoretic strengthening but a
  genuinely proof-relevant enrichment.  This enrichment is not an artefact of
  the formalism.  Theorem~\ref{thm:C} below demonstrates that the $\Kinf$ model
  separates $\beta$-witnesses from $\eta$-witnesses even when they connect the
  same pair of terms, so the additional structure carried by the proof-relevant
  tower is genuinely non-degenerate.
\end{remark}

\subsection{Higher conversions}

\begin{definition}[Explicit higher cells]
  Given terms $M, N$, a \emph{2-cell} $\alpha: p \Rightarrow q$ between
  reduction sequences $p, q : \ReductionSeq(M,N)$ is a formal witness of
  homotopy between $p$ and $q$, built inductively from reflexivity, symmetry,
  transitivity, congruence under whiskering, horizontal composition, and the
  structural associator and unitor operations.

  A \emph{3-cell} $\Theta: \alpha \Rrightarrow \beta$ between 2-cells is built
  one dimension higher, with interchange, pentagon, triangle, and the higher
  whiskering/composition witnesses appearing at this level.
\end{definition}

The tower of $n$-cells is organized into \emph{$n$-conversions} $\Sigma_n$
(reduction sequences at dimension~1, homotopies at dimension~2, etc.) and
\emph{$n$-terms} $\Pi_n$ (explicit witnesses for $n$-conversions).
Here $\Sigma_n$ denotes the packed $n$th level of the globular tower; the
boundary-indexed fibres $\ReductionSeq(M,N)$, $\mathrm{Homotopy}_2(p,q)$, and
$\mathrm{Homotopy}_3(\alpha,\beta)$ are the corresponding slices over chosen
source/target data.

\section{Baseline: Soundness and \texorpdfstring{$\THeq \subseteq \HoTFT$}{TH\_lambda= subset HoTFT}}
\label{sec:baseline}

The results of this section are not new: they were established in
Paper~I~\cite{MRdQ:theory} and are restated here in the notation of the present
paper so that the subsequent sections can cite them precisely.  The reader
already familiar with~\cite{MRdQ:theory} may safely skip to
Section~\ref{sec:tower}.

\medskip
The foundation of the semantic theory is the soundness of $\beta$- and
$\eta$-reduction, and the architecture by which that foundation lifts through
the explicit tower is worth stating precisely, both because it fixes notation
and because the subsequent sections extend the same inductive strategy to
higher-dimensional data.

The logical structure of the lifting is as follows.  At the base, one proves
that a single $\beta$-step (respectively $\eta$-step) forces equality of
semantic interpretations in any extensional Kan complex; equivalently, it
supplies the degenerate semantic $1$-cell determined by that equality.  This is
the content of the two soundness theorems below.  A key supporting ingredient is the
\emph{substitution lemma}: the semantic interpretation commutes with syntactic
substitution,
\[
  \llbracket M[N/0] \rrbracket_\rho
  \;=\;
  \llbracket M \rrbracket_{\rho[0\mapsto \llbracket N\rrbracket_\rho]},
\]
which reduces soundness of the $\beta$-rule to a computation in the function
space.  From these individual-step results one passes to composite reduction
sequences by induction on sequence length: concatenation of sequences lifts to
composition of semantic paths, and formal inversion lifts to symmetry of paths.

At the next level, each inductive constructor of a $2$-cell---left whiskering,
right whiskering, horizontal composition, associators, and unitors---is shown
to lift to a semantic $2$-cell, using the composition and whiskering operations
of the Kan complex.  The $3$-cell constructors (interchange, pentagon,
triangle, and the higher whiskering/composition witnesses) then lift by the
same scheme one dimension higher, using the $3$-dimensional horn-filling
structure of the Kan complex to produce the required higher homotopies.  In
each case the inductive step is structural: the semantic image of a composite
constructor is the corresponding composite of the semantic images of its
constituents.

The three theorems below record the cumulative outcome of this lifting
procedure.

\begin{theorem}[$\beta$-soundness]
  For every extensional Kan complex $K$, every $\beta$-step $M \to_\beta N$
  yields equality
  \[
    \llbracket M \rrbracket_\rho = \llbracket N \rrbracket_\rho
  \]
  for every environment $\rho$, and hence the corresponding degenerate semantic
  $1$-cell.
\end{theorem}

\begin{theorem}[$\eta$-soundness]
  For every extensional Kan complex $K$, every $\eta$-step $M \to_\eta N$
  yields equality
  \[
    \llbracket M \rrbracket_\rho = \llbracket N \rrbracket_\rho
  \]
  for every environment $\rho$, and hence the corresponding degenerate semantic
  $1$-cell.
\end{theorem}

By induction on the structure of reduction sequences and their higher-cell
constructors, these individual-step soundness results lift to witness-level
interpretations of the explicit tower:

\begin{theorem}[{$\THeq \subseteq \mathrm{Theory}_1(K)$,
  $\Sigma_2 \subseteq \mathrm{Theory}_2(K)$,
  $\Sigma_3 \subseteq \mathrm{Theory}_3(K)$}]
\label{thm:baseline}
  Let $K$ be an extensional Kan complex. Then:
  \begin{enumerate}[label=(\roman*)]
    \item Every reduction sequence $p : \ReductionSeq(M,N)$ induces a semantic
      path in $K$.  In particular, whenever $M =_{\beta\eta} N$ in the
      proposition-level theory $\THeq$, choosing any explicit
      reduction-sequence representative yields a proof-relevant semantic lift in
      $\mathrm{Theory}_1(K)$.
    \item If $p,q : \ReductionSeq(M,N)$ and $\alpha : p \Rightarrow q$, then
      $\alpha$ induces a semantic 2-cell between the interpreted $1$-cells of
      $p$ and $q$.
    \item If $p,q : \ReductionSeq(M,N)$, $\alpha,\beta : p \Rightarrow q$, and
      $\Theta : \alpha \Rrightarrow \beta$, then $\Theta$ induces a semantic
      3-cell between the interpreted $2$-cells of $\alpha$ and $\beta$.
  \end{enumerate}
\end{theorem}

Universal quantification over $K$ recovers the global paper-I theorem
$\THeq \subseteq \HoTFT$ together with its $2$- and $3$-dimensional companions.
These results are exactly those established in~\cite{MRdQ:theory}.  We take
them as the starting point for the new theorems presented below.

\section{Higher Conversions and the Explicit Tower}\label{sec:tower}

\subsection{The globular tower of \texorpdfstring{$n$}{n}-conversions}

The $n$-conversion tower is a \emph{globular set}: at dimension~0 one has
terms, at dimension~1 reduction sequences between terms, at dimension~2
homotopies between reduction sequences, at dimension~3 higher homotopies
between homotopies, and so on. Source and target maps $s, t : \Sigma_{n+1} \to
\Sigma_n$ record the boundary of each cell, and the globularity conditions
$s \circ s = s \circ t$ and $t \circ s = t \circ t$ hold by construction.

Through dimensions~$0$--$3$, the cells are built by explicit inductive types
carrying the full combinatorial structure of $\lambda$-reduction: whiskering,
horizontal composition, associators, pentagon and triangle coherences,
interchange, and so on. At dimension~$4$ and above, the cells are generated by
a uniform \emph{recursive completion} principle.

\subsection{The low-dimensional coherence core}

The low-dimensional structure through dimension~3 forms the core
$\mathscr{G}_\lambda$ used later for the recursive completion. Concretely, the
following data are available:

\begin{itemize}
  \item \emph{Composition}: concatenation of reduction sequences.
  \item \emph{Inverses}: formal reversal of reduction sequences.
  \item \emph{Whiskering}: left and right whiskering of 2-cells along 1-cells.
  \item \emph{Horizontal composition}: horizontal composition of 2-cells.
  \item \emph{Associator}: a 2-cell
    $\alpha_{p,q,r} : (p \cdot q) \cdot r \Rightarrow p \cdot (q \cdot r)$.
  \item \emph{Unitors}: 2-cells $\ell_p : \mathrm{id} \cdot p \Rightarrow p$
    and $\rho_p : p \cdot \mathrm{id} \Rightarrow p$.
  \item \emph{Pentagon}: a 3-cell witnessing the standard pentagon identity for
    four composable 1-cells.
  \item \emph{Triangle}: a 3-cell witnessing the triangle identity relating the
    associator and the unitors.
  \item \emph{Interchange}: 3-cells witnessing both parenthesizations of
    horizontal composition.
\end{itemize}

All whiskering operations preserve identity, composition, symmetry, and
congruence up to explicit 3-cells.

\begin{definition}[Low-dimensional $\lambda$-coherence core]
  The core $\mathscr{G}_\lambda$ is the structure
  $\mathscr{G}_\lambda$ with objects~$= \mathrm{Term}$, 1-cells~$=
  \ReductionSeq$, 2-cells~$= \mathrm{Homotopy}_2$, 3-cells~$=
  \mathrm{Homotopy}_3$, and the above operations, satisfying the globularity
  conditions at all low dimensions.
\end{definition}

\begin{definition}[Reflexive $6$-cell extension]\label{def:pack6}
  For the comparison theorem below we extend the low-dimensional core through
  dimensions~$4$, $5$, and $6$ by adjoining only the formal
  reflexive--symmetric--transitive layers generated from the $3$-cell core:
  \[
    \mathscr{G}^{\le 6}_{\lambda,4}(\eta,\theta):=\HigherDeriv(\eta,\theta),
    \qquad
    \mathscr{G}^{\le 6}_{\lambda,5}(\omega,\xi):=\HigherDeriv(\omega,\xi),
  \]
  \[
    \mathscr{G}^{\le 6}_{\lambda,6}(\mu,\nu):=\HigherDeriv(\mu,\nu),
  \]
  for parallel $3$-, $4$-, and $5$-cells respectively.  No new primitive
  coherence generators are added in these dimensions.  We refer to the resulting
  $0$--$6$ dimensional package as the \emph{reflexive $6$-cell extension} of
  $\mathscr{G}_\lambda$.
\end{definition}

\section{Comparison with the Recursive Completion}\label{sec:all-dim}

\subsection{Recursive completion of the tower}

This section has one input, one output, and one genuinely non-formal step.
The input is the explicit low-dimensional $\lambda$-conversion tower together
with the reflexive $6$-cell extension from
Definition~\ref{def:pack6}; the output is a realization map from the recursive
completion into the explicit all-dimensional tower.  The only substantive
comparison work occurs in dimensions $4$--$6$, after which the construction is
uniform.

Four objects remain distinct throughout what follows: the explicit syntactic
tower $\Sigma_\bullet$; the low-dimensional core $\mathscr{G}_\lambda$; its
recursive completion $\mathscr{R}_\bullet(\mathscr{G}_\lambda)$; and, later,
the semantic towers living in a model $K$.  Theorem~\ref{thm:A} compares
$\mathscr{R}_\bullet(\mathscr{G}_\lambda)$ with $\Sigma_\bullet$; Theorem~\ref{thm:B}
concerns the semantic side.

The weak $\omega$-groupoid $\mathscr{G}_\lambda$ constructed from explicit
$\beta\eta$-paths already gives concrete data in low dimensions:
\[
\Sigma_0=\mathrm{Term},\qquad
\Sigma_1=\ReductionSeq,\qquad
\Sigma_2=\mathrm{Homotopy}_2,\qquad
\Sigma_3=\mathrm{Homotopy}_3.
\]
Three layers are in play throughout this section: the primitive syntactic tower
$\Sigma_0$--$\Sigma_3$; the reflexive $6$-cell extension used by the generic
coherence API; and the recursive completion that continues uniformly above
dimension~$6$.  The point of the argument is to identify these layers where
they overlap and then pass to the common recursive continuation.

Moreover, the low-dimensional core also carries the next layers of reflexive
higher coherence needed to iterate the construction.  The point of this section
is therefore structural rather than model-specific: once this finite core is
fixed, no new primitive coherence operations are needed in higher dimension,
and one can continue the tower uniformly by freely adding formal reflexivity,
symmetry, and transitivity witnesses.  The non-formal part of the argument is
to show that this equality-generated continuation agrees with the explicit
low-dimensional presentation and preserves strict globular boundaries.
Figure~\ref{fig:tower-architecture} summarizes the three phases of that
comparison.

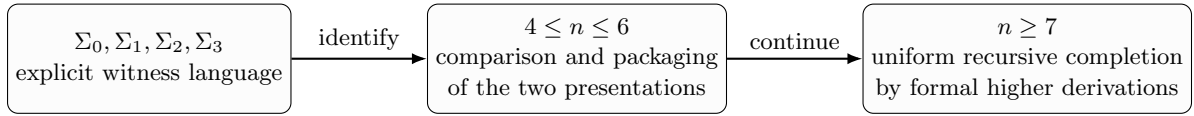
\begin{figure}[htbp]
\centering
\begin{tikzpicture}[
  phase/.style={
    draw,
    rounded corners,
    align=center,
    minimum width=3.5cm,
    minimum height=1.45cm,
    fill=gray!8
  },
  flow/.style={-{Latex[length=2mm]}, thick},
  note/.style={align=center}
]
  \node[phase] (explicit) {$\Sigma_0,\Sigma_1,\Sigma_2,\Sigma_3$\\explicit witness language};
  \node[phase, right=1.8cm of explicit] (bridge) {$4 \le n \le 6$\\comparison and packaging\\of the two presentations};
  \node[phase, right=1.8cm of bridge] (recursive) {$n \ge 7$\\uniform recursive completion\\by formal higher derivations};

  \draw[flow] (explicit) -- node[above] {identify} (bridge);
  \draw[flow] (bridge) -- node[above] {continue} (recursive);

  \node[note, below=0.95cm of bridge] {$\mathscr{R}_\bullet(\mathscr{G}_\lambda)$ admits a strict boundary-preserving realization\\inside the explicit tower $\Sigma_\bullet$.};
\end{tikzpicture}
\caption{Architecture of the proof of Theorem~\ref{thm:A}.  The explicit low-dimensional tower is first matched with the equality-generated continuation in the bridge dimensions, after which the recursive completion proceeds uniformly.}
\label{fig:tower-architecture}
\end{figure}

\begin{remark}[Which layers are explicit]
Two different presentations are in play here.  The syntactic tower
$\Sigma_\bullet$ is explicit through dimension~$3$, and for $n\ge 4$ its cells
are already given recursively as higher-derivation data between parallel
$(n-1)$-cells.  Separately, the shared simplicial core used by the generic
coherence packaging is a reflexive $6$-cell core: its $4$-, $5$-, and $6$-cells
are packaged reflexivity layers attached to the low-dimensional
$\lambda$-coherence interface.  The role of dimensions $4$--$6$ in
Theorem~\ref{thm:A} is precisely to identify these two presentations before the
uniform recursive continuation begins above dimension~$6$.
\end{remark}

\begin{definition}[Higher derivation]\label{def:higherderiv}
Let $X$ be a type.  The type $\HigherDeriv_X(x,y)$ of \emph{higher
derivations} between $x,y\in X$ is generated inductively by
\[
\mathsf{refl}_x : \HigherDeriv_X(x,x),\qquad
\mathsf{symm} : \HigherDeriv_X(x,y)\to \HigherDeriv_X(y,x),
\]
and
\[
\mathsf{trans} : \HigherDeriv_X(x,y)\to \HigherDeriv_X(y,z)\to
\HigherDeriv_X(x,z).
\]
Thus $\HigherDeriv_X$ is the reflexive--symmetric--transitive closure of the
equality data already present on $X$.
\end{definition}

The type $\HigherDeriv_X(x,y)$ does not introduce new primitive generators
between unrelated objects.  In particular, any element of
$\HigherDeriv_X(x,y)$ determines that
the endpoints agree.  Accordingly, above the explicit core the recursive tower
is \emph{equality-generated}: higher cells are added by iterating formal
reflexivity, symmetry, and transitivity on previously constructed parallel
cells.  Its central structural property is still \emph{functoriality}: every
map on the underlying type transports higher derivations covariantly.

\begin{proposition}[Functoriality of $\HigherDeriv$]\label{prop:hd-functor}
  Let $f:X\to Y$ be any map.  There is a canonical map
  \[
    \HigherDeriv(f) : \HigherDeriv_X(x,y) \longrightarrow
    \HigherDeriv_Y\bigl(f(x),\,f(y)\bigr),
  \]
  defined by recursion on the constructors of\/ $\HigherDeriv$:
  \begin{align*}
    \HigherDeriv(f)(\mathsf{refl}_x)
      &= \mathsf{refl}_{f(x)},\\
    \HigherDeriv(f)(\mathsf{symm}(h))
      &= \mathsf{symm}\bigl(\HigherDeriv(f)(h)\bigr),\\
    \HigherDeriv(f)(\mathsf{trans}(h_1,h_2))
      &= \mathsf{trans}\bigl(\HigherDeriv(f)(h_1),\,
         \HigherDeriv(f)(h_2)\bigr).
  \end{align*}
  Moreover, the assignment $f\mapsto \HigherDeriv(f)$ respects identity and
  composition:
  \[
    \HigherDeriv(\mathrm{id}_X) = \mathrm{id},\qquad
    \HigherDeriv(g\circ f) = \HigherDeriv(g)\circ \HigherDeriv(f).
  \]
\end{proposition}

\begin{proof}
  The recursive clauses define $\HigherDeriv(f)$ by structural induction on
  derivation terms.  The identity law is immediate:
  $\HigherDeriv(\mathrm{id})$ fixes every constructor pointwise, since
  $\mathsf{refl}_x\mapsto \mathsf{refl}_{\mathrm{id}(x)}=\mathsf{refl}_x$,
  $\mathsf{symm}(h)\mapsto \mathsf{symm}(\HigherDeriv(\mathrm{id})(h))$,
  and the induction hypothesis closes the recursion.

  For the composition law, the base case is:
  \[
    \begin{aligned}
      \HigherDeriv(g)(\HigherDeriv(f)(\mathsf{refl}_x))
        &= \HigherDeriv(g)(\mathsf{refl}_{f(x)}) \\
        &= \mathsf{refl}_{g(f(x))} \\
        &= \HigherDeriv(g\circ f)(\mathsf{refl}_x).
    \end{aligned}
  \]
  The inductive cases for $\mathsf{symm}$ and $\mathsf{trans}$ follow by the
  same structural unwinding: the outer map $\HigherDeriv(g)$ distributes
  across the constructors produced by $\HigherDeriv(f)$, and the induction
  hypothesis matches each sub-derivation.
\end{proof}

Proposition~\ref{prop:hd-functor} is the engine that drives the realization
theorem below: at each dimension, the realization of a higher derivation is
obtained by applying $\HigherDeriv$ to the endpoint realization map.  The
identity and composition laws ensure that this process is compatible with the
globular boundary structure and with the recursive definition of source and
target.

Starting from the reflexive $6$-cell extension of $\mathscr{G}_\lambda$, we
define a recursive globular tower $\mathscr{R}_\bullet(\mathscr{G}_\lambda)$ as
follows.
For dimensions $0,\dots,6$ we keep the given low-dimensional cells of the core
explicitly.  For every $n\ge 6$, an $(n+1)$-cell is a triple
\[
(x,y,h),\qquad x,y\in \mathscr{R}_n(\mathscr{G}_\lambda),\quad
h\in \HigherDeriv(x,y).
\]
Its source is $x$ and its target is $y$.  The globularity equalities are then
automatic: if $h$ is a higher derivation from $x$ to $y$, the two iterated
sources and the two iterated targets agree because every constructor of
$\HigherDeriv$ preserves equality of endpoints.

This recursive completion should be compared with the explicit tower
$\Sigma_\bullet$ of higher $\lambda$-conversions.  The latter is already
constructive in all dimensions, but its presentation is not literally the same
as the generic recursive completion in dimensions $4,5,6$: there one must
package the previously constructed $3$-cells, $4$-cells, and $5$-cells as
explicit cells of the target tower.  Once these dimensions are handled, however,
the higher stages on both sides are governed by the same rule: ``a higher cell
is a higher derivation between parallel cells of the preceding dimension.''

The interface between the explicit low-dimensional tower and the recursive
completion in the intermediate dimensions is captured by the following
observation, which we state as a separate lemma because it isolates the only
non-trivial step in passing from the low-dimensional core to the recursive
continuation.

\begin{lemma}[Low-dimensional packaging]\label{lem:packaging}
  For $d\in\{3,4,5\}$, let $\Sigma_d$ denote the explicit $d$-cells of the
  $\lambda$-$\omega$-groupoid, and let $\mathscr{R}_{d+1}(\mathscr{G}_\lambda)$
  denote the $(d+1)$-cells of the recursive completion.  There is a canonical
  \emph{packaging map}
  \[
    \mathsf{pack}_{d+1}:\mathscr{R}_{d+1}(\mathscr{G}_\lambda)
    \longrightarrow \Sigma_{d+1}
  \]
  defined by:
  \begin{enumerate}[label=(\alph*)]
    \item embedding the two endpoint $d$-cells into $\Sigma_d$ via the identity
      (for $d\le 3$) or via the inductively defined packaging map
      $\mathsf{pack}_d$ (for $d\in\{4,5\}$), and then
    \item transporting the higher derivation datum along the resulting endpoint
      map, using the functoriality of $\HigherDeriv$
      (Proposition~\textup{\ref{prop:hd-functor}}).
  \end{enumerate}
  No new primitive coherence operations beyond those already present in the
  $3$-cell core enter the construction.
\end{lemma}

\begin{proof}
  At dimension $4$, the endpoints are $3$-cells, which belong to the explicit
  core and therefore need no packaging: $\mathsf{pack}_4$ applies
  $\HigherDeriv(\mathrm{id})$ to the derivation datum, which by the identity law
  of Proposition~\ref{prop:hd-functor} is the identity.  Concretely, a $4$-cell
  of the recursive completion consists of two parallel $3$-cells $\eta,\theta$
  together with a higher derivation $\omega:\HigherDeriv(\eta,\theta)$, and
  $\mathsf{pack}_4$ sends this to the explicit $4$-cell $(\eta,\theta,\omega)$
  verbatim.

  At dimension $5$, the endpoints are $4$-cells.  Each such $4$-cell was itself
  produced by $\mathsf{pack}_4$, so $\mathsf{pack}_5$ applies
  $\HigherDeriv(\mathsf{pack}_4)$ to the derivation data.  The composition law
  of Proposition~\ref{prop:hd-functor} guarantees that iterated packaging
  respects the boundary maps.  The argument for dimension $6$ is identical, with
  $\mathsf{pack}_5$ replacing $\mathsf{pack}_4$.

  The key observation is that every constructor invoked---$\mathsf{refl}$,
  $\mathsf{symm}$, $\mathsf{trans}$, and the endpoint maps---already exists in
  the $3$-cell core.  The packaging maps merely \emph{compose} these existing
  operations through the functoriality of $\HigherDeriv$; they do not introduce
  any new generating data.
\end{proof}

\subsection{The realization theorem}

\begin{definition}[Boundary-preserving realization package]
Let $\mathscr{C}$ be a reflexive low-dimensional globular core and let
$\mathscr{T}$ be a globular tower.  A \emph{boundary-preserving realization package}
from $\mathscr{C}$ to $\mathscr{T}$ consists of maps
\[
\realize_n:\mathscr{R}_n(\mathscr{C})\longrightarrow \mathscr{T}_n
\qquad (n\in\mathbb N),
\]
where $\mathscr{R}_\bullet(\mathscr{C})$ is the recursive completion of
$\mathscr{C}$, such that for every $(n+1)$-cell $x$,
\[
s(\realize_{n+1}(x))=\realize_n(sx),
\qquad
t(\realize_{n+1}(x))=\realize_n(tx).
\]
\end{definition}

\begin{theorem}[Comparison with the recursive completion]
\label{thm:A}
The explicit higher $\lambda$-conversion tower $\Sigma_\bullet$ admits an
boundary-preserving realization package from the recursive completion of the
low-dimensional core $\mathscr{G}_\lambda$.  In particular:
\begin{enumerate}[label=(\roman*)]
  \item For every $n\in\mathbb N$, there is a realization map
    \[
      \realize_n:\mathscr{R}_n(\mathscr{G}_\lambda)\to \Sigma_n.
    \]
  \item These realization maps commute strictly with source and target maps.
  \item Hence $\Sigma_\bullet$ carries a strict boundary-preserving realization
    of $\mathscr{R}_\bullet(\mathscr{G}_\lambda)$ extending the low-dimensional
    core.
\end{enumerate}
\end{theorem}

\begin{proof}
We divide the construction into the explicit core, the finite packaging phase,
and the uniform recursive continuation.

\smallskip
\noindent
\emph{Step 1: dimensions $0$--$3$.}
At the first four levels there is nothing to prove beyond unpacking the
definitions.  The objects of $\mathscr{G}_\lambda$ are $\lambda$-terms, its
$1$-cells are explicit reduction zigzags, its $2$-cells are explicit
homotopies between parallel zigzags, and its $3$-cells are explicit homotopies
between parallel $2$-cells.  Thus
\[
\realize_0,\realize_1,\realize_2,\realize_3
\]
are simply the identity identifications between the recursively completed tower
and the explicit tower.

\smallskip
\noindent
\emph{Step 2: dimensions $4$--$6$.}
By Definition~\ref{def:pack6}, the intermediate dimensions are already the
formal reflexive layers generated from the $3$-cell core.
Lemma~\ref{lem:packaging} therefore provides canonical packaging maps
\[
\mathsf{pack}_4,\qquad \mathsf{pack}_5,\qquad \mathsf{pack}_6
\]
from the recursive side to the corresponding explicit cells of
$\Sigma_4,\Sigma_5,\Sigma_6$.  We define
\[
\realize_4:=\mathsf{pack}_4,\qquad
\realize_5:=\mathsf{pack}_5,\qquad
\realize_6:=\mathsf{pack}_6.
\]
The source/target commutation in these dimensions is exactly the content of
Lemma~\ref{lem:packaging} together with the functoriality of $\HigherDeriv$ from
Proposition~\ref{prop:hd-functor}: packaging an endpoint first and then taking
its boundary agrees with taking the boundary in the reflexive $6$-cell
extension first and then packaging the result.

\smallskip
\noindent
\emph{Step 3: recursive continuation above dimension $6$.}
Suppose now that $\realize_n$ has already been defined for some $n\ge 6$.
An $(n+1)$-cell of the recursive completion is, by definition, a triple
$(x,y,h)$ with $x,y\in \mathscr{R}_n(\mathscr{G}_\lambda)$ and
$h\in \HigherDeriv(x,y)$.  We define
\[
\realize_{n+1}(x,y,h)
:=
\bigl(\realize_n(x),\,\realize_n(y),\,
\HigherDeriv(\realize_n)(h)\bigr).
\]
Here $\HigherDeriv(\realize_n)(h)$ denotes the image of $h$ under the map
obtained by recursion on the constructors of $\HigherDeriv$:
\[
\mathsf{refl}(x)\mapsto \mathsf{refl}(\realize_n(x)),\qquad
\mathsf{symm}(h)\mapsto \mathsf{symm}(\HigherDeriv(\realize_n)(h)),
\]
\[
\mathsf{trans}(h_1,h_2)\mapsto
\mathsf{trans}(\HigherDeriv(\realize_n)(h_1),
               \HigherDeriv(\realize_n)(h_2)).
\]
Thus the realization in dimension $n+1$ is forced by the realization in
dimension $n$.

The conceptual point is worth highlighting: no new coherence operations,
packaging maps, or case analyses enter the construction above dimension~$6$.
The reason is structural.  Below dimension~$4$, the explicit tower is defined by
a bespoke inductive type carrying whiskering, interchange, pentagon, and
triangle constructors.  In dimensions $4$--$6$, one must translate between that
  bespoke presentation and the generic recursive completion, which requires the
  packaging maps of Lemma~\ref{lem:packaging}.  But from dimension~$7$ onward,
  \emph{both} towers are defined by the same rule: a cell is a higher derivation
  between parallel cells of the preceding dimension.  Above that range, the
  realization map becomes uniform: it is given by the same
  higher-derivation recursion on both sides, so no new comparison data are
  needed beyond the inductive hypothesis.  In
  particular, no finite list of ``exceptional dimensions'' grows as one climbs the
  tower: the base and packaging phases are completed once and for all in
  dimensions $0$--$6$, and the rest is pure induction.

\smallskip
\noindent
\emph{Step 4: compatibility with source and target.}
For dimensions $0$--$6$, compatibility is checked directly from the formulas:
both the recursive tower and the explicit $\lambda$-tower compute source by
forgetting the last boundary component, and compute target by forgetting the
corresponding source component.

For dimensions above $6$, let $x=(u,v,h)$ be an $(n+1)$-cell.  Then, by
definition of the recursive source and target maps,
\[
s(x)=u,\qquad t(x)=v.
\]
Applying the recursive formula for realization gives
\[
\realize_{n+1}(x)=
\bigl(\realize_n(u),\realize_n(v),\HigherDeriv(\realize_n)(h)\bigr).
\]
Hence its source in the explicit tower is exactly $\realize_n(u)$ and its target
is exactly $\realize_n(v)$, so
\[
s(\realize_{n+1}(x))=\realize_n(sx),\qquad
t(\realize_{n+1}(x))=\realize_n(tx).
\]
This proves strict source and target compatibility in all dimensions.

\smallskip
\noindent
\emph{Step 5: conclusion.}
The realization maps realize the recursively completed tower of the
$\lambda$-coherence core inside the explicit higher
$\lambda$-conversion tower at every dimension.  Since the recursive completion
is globular by construction and the realization preserves boundaries strictly,
the explicit tower inherits the boundary-preserving realization package of the
theorem statement.
\end{proof}

\subsection{Truncation and the classical theory}

An immediate corollary of Theorem~\ref{thm:A} is that the recursive completion
extends, rather than replaces, the ordinary theory of $\beta\eta$-equality.

\begin{corollary}[0-truncation]
The $0$-truncation of the low-dimensional core---and therefore also of its
recursive completion---recovers the classical $\beta\eta$-equality:
\[
\pi_0(\mathscr{G}_\lambda)(M,N)
\;\Longleftrightarrow\;
M =_{\beta\eta} N.
\]
\end{corollary}

\begin{proof}
The proposition $\pi_0(\mathscr{G}_\lambda)(M,N)$ asserts only that there exists
a $1$-cell from $M$ to $N$, that is, an explicit reduction sequence
$p:\ReductionSeq(M,N)$.  Forgetting the witness $p$ yields an ordinary
$\beta\eta$-conversion from $M$ to $N$.  Conversely, every ordinary
$\beta\eta$-conversion admits an explicit reduction-sequence representative.
Thus existence of a $1$-cell in the recursive completion---and hence in its
realization inside the explicit tower---is equivalent to the classical
proposition-level equality relation.
\end{proof}

\section{Sufficient Semantic Coherence Data}\label{sec:front-seed}

This section asks a narrow semantic question: how much chosen coherence data
are actually needed for the later associator and pentagon arguments?  The input
is an extensional Kan complex equipped with a specific small front-facing
package of $3$-cells; the output is the recursive associator comparison, the
semantic pentagon comparison on interpreted reduction sequences, and the bridge
theorems used later.

\subsection{Three semantic regimes}

The semantic development separates three logically distinct layers.

\begin{enumerate}[label=(\arabic*)]
  \item \textbf{Bare semantics.}
    An extensional Kan complex suffices for the interpretation of terms,
    $\beta$- and $\eta$-soundness, the interpretation of explicit reduction
    sequences, and the semantic images of the explicit $2$- and $3$-cells
    already present in the syntax.

  \item \textbf{Coherent semantics.}
    A coherent extensional Kan complex adds explicit $3$-dimensional comparison
    data for associativity and whiskering.  This makes it possible to compare
    different semantic parenthesizations of composite interpreted paths.

  \item \textbf{Strict semantics.}
    In a strict extensional Kan complex, the required coherent comparison data
    are derived axiom-freely from uniqueness of horn fillers, so no additional
    chosen coherence fields are needed beyond the extensional $\lambda$-model
    structure.
\end{enumerate}

We now show that a smaller front-facing seed already generates the full
associator and pentagon package needed later.

\subsection{The front-seed boundary data}

\begin{definition}[Front-seed coherent extensional Kan complex]
A \emph{front-seed coherent extensional Kan complex} is an extensional Kan
complex $K$ equipped with the following additional $3$-cell data.

\begin{enumerate}[label=(FS\arabic*)]
  \item \textbf{WLWR comparison.}
    For every composable $1$-cells
    \[
    a \xrightarrow{\alpha} b \xrightarrow{\beta,\gamma} c \xrightarrow{\delta} d
    \]
    and every $2$-cell $\eta:\beta\Rightarrow\gamma$, a chosen $3$-cell
    \[
    \whiskerRight(\whiskerLeft(\alpha,\eta),\delta)
    \;\Longrightarrow\;
    \ass(\alpha,\beta,\delta)\;\cdot\;
    \whiskerLeft(\alpha,\whiskerRight(\eta,\delta))\;\cdot\;
    \ass(\alpha,\gamma,\delta)^{-1}.
    \]
    In words: right-whiskering a left-whiskered $2$-cell agrees, up to the
    appropriate associator corrections, with first right-whiskering and then
    left-whiskering.

  \item \textbf{Inner-right-front pentagon contraction.}
    For every composable quadruple $p,q,r,s$, a chosen $3$-cell contracting the
    inner-right-front face of the semantic pentagon to the reflexive $2$-cell on
    the edge labeled by $q\cdot(r\cdot s)$ before left-whiskering by $p$.
\end{enumerate}
\end{definition}

The second datum is deliberately asymmetric: it does not assert the entire
pentagon at once, but only contracts one distinguished front face.  The theorem
below says that, together with WLWR, this already forces the higher coherence
needed for associator comparison and for the full pentagon at interpreted
reduction sequences.

\medskip
It is useful to describe the geometric content of the two seeds more precisely.

\smallskip
\noindent
\textbf{FS1 (WLWR comparison): front-face horn filling for mixed whiskering.}\;
Given a $2$-cell $\eta:\beta\Rightarrow\gamma$ flanked by $1$-cells $\alpha$ on
the left and $\delta$ on the right, there are two natural ways to extend $\eta$
to the full composite boundary: first left-whisker by $\alpha$ and then
right-whisker by $\delta$, or first right-whisker by $\delta$ and then
left-whisker by $\alpha$.  These two expressions land at the same pair of
endpoints $\alpha\cdot\beta\cdot\delta$ and $\alpha\cdot\gamma\cdot\delta$, but
they are not \emph{a priori} the same $2$-cell because the whiskering operations
interact with the associativity of path composition.  Seed~(FS1) asserts that
the discrepancy is controlled by the associator: the two mixed whiskerings
differ by exactly the pair of associator $2$-cells that rebrackets the
composite.  In simplicial terms, FS1 fills a specific $3$-horn whose two
front faces are the two mixed-whiskering composites and whose remaining faces
are associator instances.

\smallskip
\noindent
\textbf{FS2 (inner-right-front pentagon contraction): partial boundary
seeding.}\;
The pentagon coherence for four composable $1$-cells $p,q,r,s$ is a $3$-cell
whose boundary consists of five associator $2$-cells arranged around the
boundary of a $4$-simplex.  Rather than positing the full pentagon $3$-cell as
primitive, FS2 provides only the contraction of the \emph{inner-right-front}
face---the face built from the associator $\ass(q,r,s)$ whiskered by $p$ on
the left.  This contracts a particular sub-composite of the pentagon boundary
to the identity $2$-cell on the edge labeled by $q\cdot(r\cdot s)$ before
left-whiskering by $p$.  Geometrically, FS2 is a partial horn-filling datum:
it specifies one non-degenerate face of the pentagon $4$-horn, from which the
remaining faces can be assembled using the WLWR comparison of FS1 together with
the standard Kan-complex horn-filling machinery.

\begin{definition}[Structural pentagon data]\label{def:structural-pentagon}
For composable reduction sequences $p,q,r,s$, write
\[
P_0:=(((p\cdot q)\cdot r)\cdot s),\qquad
P_1:=((p\cdot q)\cdot(r\cdot s)),\qquad
P_2:=((p\cdot(q\cdot r))\cdot s),
\]
\[
P_3:=p\cdot((q\cdot r)\cdot s),\qquad
P_4:=p\cdot(q\cdot(r\cdot s)).
\]
The two $2$-dimensional boundary composites of the structural pentagon are
\[
L_{p,q,r,s}:=\ass(p\cdot q,r,s)\;\cdot\;\ass(p,q,r\cdot s)
\quad\text{and}\quad
R_{p,q,r,s}:=(\ass(p,q,r)\ast s)\;\cdot\;\ass(p,q\cdot r,s)\;\cdot\;(p\ast\ass(q,r,s)),
\]
both regarded as $2$-cells from $P_0$ to $P_4$.  We call the partial
$4$-dimensional boundary built from the FS2 inner-right-front contraction, the
two WLWR comparison faces, and the remaining back associator face the
\emph{structural pentagon horn} on $p,q,r,s$; its missing face is the desired
$3$-cell $L_{p,q,r,s}\Rightarrow R_{p,q,r,s}$.
\end{definition}

\medskip
The significance of Theorem~\ref{thm:B} is a precise reduction of the semantic
coherence package actually used later.  In the formal development, the larger
coherent interface with a full pentagon witness and the smaller front-seed
interface are both kept explicitly.  The front-seed theorem shows that the
  smaller chosen-data interface already suffices: the full pentagon field may be
   replaced by the inner-right-front contraction, from which one derives the
   recursive associator comparison, the structural semantic pentagon comparison on
   interpreted reduction sequences, and the source/target/shell bridge theorems
   for the explicit syntactic pentagon, with the target and shell bridges landing
   in the mixed target shell determined by the front-seed associator comparison.
   Thus the gain is not an appeal to an unspecified larger
   coherence package, but a concrete reduction from a full pentagon datum to a
   smaller front-facing seed, while keeping WLWR explicit.

\begin{lemma}[WLWR head-step normalization]\label{lem:frontseed-wlwr-normalization}
Let $K$ be a front-seed coherent extensional Kan complex, let $\alpha$ be the
semantic image of a generating head step, let $\eta:\beta\Rightarrow\gamma$ be
a semantic $2$-cell, and let $\delta$ be a further semantic edge.  Then the
mixed whiskering composite
\[
\whiskerRight(\whiskerLeft(\alpha,\eta),\delta)
\]
is connected by a chosen semantic $3$-cell to
\[
\ass(\alpha,\beta,\delta)\;\cdot\;
\whiskerLeft(\alpha,\whiskerRight(\eta,\delta))\;\cdot\;
\ass(\alpha,\gamma,\delta)^{-1}.
\]
\end{lemma}

\begin{proof}
This is exactly the WLWR comparison supplied by Seed~(FS1), recorded here as a
standalone normalization lemma because it is the only genuinely non-formal
head-step input used in the associator and pentagon arguments below.
\end{proof}

\begin{lemma}[Triple associator comparison]\label{lem:frontseed-assoc}
Let $K$ be a front-seed coherent extensional Kan complex and let $p,q,r$ be
reduction sequences.  Then the structural associator shell of the interpreted
composites $((p\cdot q)\cdot r)$ and $p\cdot(q\cdot r)$ is connected by a
semantic $3$-cell to the equality-generated associator comparison coming from
literal concatenation.
\end{lemma}

\begin{proof}
Proceed by induction on the first reduction sequence $p$.

\smallskip
\noindent
\emph{Base case.} If $p=\mathsf{refl}$, then the structural and
equality-generated associator comparisons coincide, so the required $3$-cell is
reflexive.

\smallskip
\noindent
\emph{Head-step case.} Suppose $p$ begins with a generating step whose semantic
image is $\alpha$, and write the remaining tail comparison as
$\eta:\beta\Rightarrow\gamma$.  The only non-formal obstruction is the nested
whiskering term
\[
\whiskerRight(\whiskerLeft(\alpha,\eta),\delta),
\]
where $\delta$ is the semantic image of $r$.  By
Lemma~\ref{lem:frontseed-wlwr-normalization}, this term is replaced by the
normalized composite
\[
\ass(\alpha,\beta,\delta)\;\cdot\;
\whiskerLeft(\alpha,\whiskerRight(\eta,\delta))\;\cdot\;
\ass(\alpha,\gamma,\delta)^{-1}.
\]
After this explicit head-step normalization, the boundary matches the structural
associator shell.

\smallskip
\noindent
\emph{Inductive step.} Write $p$ as its head step followed by a shorter tail.
Apply the WLWR normalization above to the head step and the induction
hypothesis to the tail.  Vertical composition of the resulting $3$-cells
produces a comparison from the structural associator shell of $(p,q,r)$ to the
equality-generated associator.
\end{proof}

\begin{lemma}[Front-seed pentagon horn completion]\label{lem:frontseed-pentagon}
Let $K$ be a front-seed coherent extensional Kan complex and let $p,q,r,s$ be
reduction sequences.  Then the structural pentagon horn of
Definition~\ref{def:structural-pentagon} is a well-defined $4$-dimensional horn
in $K$.
\end{lemma}

\begin{proof}
The four prescribed faces are:
\begin{enumerate}[label=(\alph*)]
  \item the inner-right-front face supplied by Seed~(FS2),
  \item the two front mixed-whiskering faces normalized by Seed~(FS1),
  \item the remaining back face built from the structural associator shell.
\end{enumerate}
The inner-right-front face has target the reflexive $2$-cell on $P_4$.  Each
front WLWR face has as one boundary edge the same mixed-whiskering composite as
the adjacent FS2 or associator face, and
Lemma~\ref{lem:frontseed-assoc} identifies the remaining boundary edges of those
WLWR faces with the corresponding associator composites built from
$L_{p,q,r,s}$ and $R_{p,q,r,s}$.  The back face is already the structural
associator shell on the same parenthesization data.  Consequently every common
$2$-dimensional overlap is described by the same composite of associators and
whiskerings when read from either adjacent $3$-face.  Hence the four prescribed
faces glue to a well-defined $4$-horn in $K$.
\end{proof}

\begin{lemma}[Source-side associator bridge]\label{lem:frontseed-source-bridge}
Let $K$ be a front-seed coherent extensional Kan complex and let $p,q,r,s$ be
reduction sequences.  The interpreted source boundary of the explicit syntactic
pentagon is connected by a semantic $3$-cell to the structural semantic source
shell determined by $L_{p,q,r,s}$.
\end{lemma}

\begin{proof}
Apply Lemma~\ref{lem:frontseed-assoc} to the two associator pieces on the source
side of the explicit pentagon.  Their composite identifies the interpreted
syntactic source boundary with the structural source shell.
\end{proof}

\begin{lemma}[Target-side associator bridge]\label{lem:frontseed-target-bridge}
Let $K$ be a front-seed coherent extensional Kan complex and let $p,q,r,s$ be
reduction sequences.  The interpreted target boundary of the explicit syntactic
pentagon is connected by a semantic $3$-cell to the mixed target shell
determined by $R_{p,q,r,s}$.
\end{lemma}

\begin{proof}
Apply Lemma~\ref{lem:frontseed-assoc} to the three associator pieces on the
target side of the explicit pentagon.  Their composite identifies the
interpreted syntactic target boundary with the mixed target shell.
\end{proof}

\begin{theorem}[Front-seed sufficiency]
\label{thm:B}
Let $K$ be a front-seed coherent extensional Kan complex.  Then:
\begin{enumerate}[label=(\roman*)]
  \item For every triple of reduction sequences $p,q,r$, there is a semantic
    $3$-cell from the structural associator shell of the interpreted composites
    $((p\cdot q)\cdot r)$ and $p\cdot(q\cdot r)$ to the
    equality-generated associator comparison coming from literal concatenation.
    Consequently any two parenthesizations of a finite composite are connected
    by a chain of such comparison $3$-cells.
  \item For every quadruple of reduction sequences $p,q,r,s$, there is a
    semantic $3$-cell
    \[
    \ass(p\cdot q,r,s)\;\cdot\;\ass(p,q,r\cdot s)
    \;\Longrightarrow\;
    (\ass(p,q,r)\ast s)\;\cdot\;\ass(p,q\cdot r,s)\;\cdot\;(p\ast\ass(q,r,s)),
    \]
    where $\ast$ denotes whiskering.
  \item For every quadruple $p,q,r,s$, the interpreted explicit syntactic
    pentagon admits a source bridge to the structural semantic source shell, a
    target bridge to the mixed target shell determined by~(i), and the
    resulting shell bridge between those two shells.
\end{enumerate}
\end{theorem}

Clause~(ii) gives the fully structural semantic pentagon comparison, whereas
clause~(iii) bridges the explicit syntactic pentagon only to the mixed target
shell arising from the front-seed associator comparison.

\begin{proof}
We prove the three claims in order.

\smallskip
\noindent
\emph{Step 1: from WLWR to the recursive associator comparison.}
Lemma~\ref{lem:frontseed-assoc} gives the comparison for every triple
$p,q,r$.  Iterating that comparison gives the statement for arbitrary
parenthesized composites of any finite list $p_1,\dots,p_k$: one changes
parenthesization one associator at a time, and each elementary move is
controlled by the triple comparison lemma.

\smallskip
\noindent
\emph{Step 2: from the front pentagon seed to the full semantic pentagon.}
By Lemma~\ref{lem:frontseed-pentagon}, the four known faces assemble to the
structural pentagon horn of Definition~\ref{def:structural-pentagon}.  In the
underlying simplicial set of $K$, this boundary is represented by the specific
\((2,1)\)-horn used in the formal proof: face~$0$ is the FS2 inner-right-front
contraction, face~$2$ is the composition triangle for $p$ against the already
parenthesized tail $q\cdot(r\cdot s)$, and face~$3$ is the composite of the two
FS1-normalized mixed-whiskering faces with the back associator shell.  The
missing face~$1$ is therefore exactly the desired pentagon comparison
$L_{p,q,r,s}\Rightarrow R_{p,q,r,s}$, and the Kan filling property supplies it.

\smallskip
\noindent
\emph{Step 3: bridge theorems from syntax to semantic shells.}
The explicit syntactic higher calculus already contains an associator $2$-cell
and a pentagon $3$-cell between explicit reduction sequences.  These objects can
be interpreted in any bare extensional Kan complex, so in particular in our
front-seed model $K$.  However, their semantic images land a priori in the
``syntactic'' presentation of the boundary, while Step~1 and Step~2 produce the
structural source shell and the mixed target shell built from the Kan-complex
associator and whiskering operations.

Lemma~\ref{lem:frontseed-source-bridge} supplies the source-side bridge and
Lemma~\ref{lem:frontseed-target-bridge} supplies the target-side bridge.
Finally, compose
\begin{enumerate}[label=(\alph*)]
  \item the source-side bridge,
  \item the interpreted syntactic pentagon $3$-cell, and
  \item the inverse of the target-side bridge.
\end{enumerate}
The result is a $3$-cell from the structural source shell to the mixed target
shell.  This is exactly the front-seed shell bridge theorem.  Thus the explicit
syntactic pentagon and the front-seed semantic pentagon boundary are not merely
parallel facts: they are connected by explicit comparison $3$-cells, while
Step~2 separately provides the full structural semantic pentagon comparison on
interpreted reduction sequences.

This proves (i), (ii), and (iii).
\end{proof}

\begin{remark}[Open strengthening target]
The remaining conceptual question is whether the two front-seed data can be
derived from the bare extensional Kan interface alone.  A positive answer would
show that no additional semantic coherence fields are needed even for the
  associator and pentagon package.  At present, what is proved is sharper but more
  modest: the full coherent interface is unnecessary, and one specific smaller
  boundary package is already sufficient, namely the WLWR comparison together
  with the inner-right-front pentagon contraction.
\end{remark}

\section{The \texorpdfstring{$\Kinf$}{K-infinity} Model}\label{sec:Kinfty}

The $\Kinf$ model was introduced in Paper~II~\cite{MRdQ:Kinfinity}, which
defined the inverse-limit construction, proved algebraicity and bounded
completeness (Proposition~4.1), and exhibited the separation phenomenon of
Example~\ref{ex:witness-span}.  The aspects of that baseline used here are
restated in this section and in Section~\ref{sec:witness}.  The main new point
of the present section is not the abstract fact that a projection-pair inverse
limit can be reflexive---a classical theme already present in Scott-style
semantics and later domain-theoretic accounts~\cite{Scott:continuous,AbramskyJung:domain}---but the
\emph{exact reflexive packaging}: explicit, globally continuous reify and
reflect maps together with a stagewise application formula that pins down the
relationship between the finite stages and the limit.

\subsection{Construction}

The object $\Kinf$ is obtained as the inverse limit of the reflexive tower
\[
  K_0,\ K_1,\ K_2,\ \ldots,
  \qquad
  K_0=N^+,\qquad K_{n+1}=[K_n\to K_n].
\]
At each stage there is a projection pair
\[
  f_n^+ : K_n \longrightarrow K_{n+1},
  \qquad
  f_n^- : K_{n+1} \longrightarrow K_n
\]
with
\[
  f_n^- \circ f_n^+ = \mathrm{id}_{K_n},
  \qquad
  f_n^+\bigl(f_n^-(u)\bigr)\leq u
  \quad (u\in K_{n+1}).
\]
Thus every successor stage is a reflexive enlargement of the previous one.

Here and below, $u\simeq v$ means mutual approximation
$u\le v$ and $v\le u$ in the underlying homotopy partial order.  We use
$\simeq$ only for inverse-limit coherence and density statements; the exact
formulas in Theorem~\ref{thm:D} are literal equalities of continuous maps and
coordinates.

\begin{definition}[The tower $K_0,K_1,\ldots$]
  Set $K_0=N^+$, the flat domain on the non-trivial homotopy-group carrier
  $N=\bigsqcup_{n\ge 0} S^n$ together with bottom, and define inductively
  \[
    K_{n+1}=[K_n\to K_n].
  \]
  The transition maps are the projection pairs $(f_n^+,f_n^-)$ above.
\end{definition}

\begin{definition}[$\Kinf$]
  The inverse limit $\Kinf$ is the set of coherent threads
  \[
    \Kinf=\varprojlim_n K_n
    =\Bigl\{(x_n)_{n\ge 0}\in \prod_{n\ge 0}K_n :
      f_n^-(x_{n+1})\simeq x_n\ \text{for all }n\Bigr\},
  \]
  ordered coordinatewise. We write
  \[
    \pi_n:\Kinf\to K_n
  \]
  for the $n$th projection.
\end{definition}

\begin{definition}[Canonical stage embedding]\label{def:stage-embedding}
  For $u\in K_n$, define the thread $f_{n,\infty}(u)\in\Kinf$ by its
  coordinates
  \[
    \pi_m(f_{n,\infty}(u))=
    \begin{cases}
      f_m^- \circ \cdots \circ f_{n-1}^-(u), & m<n,\\[1ex]
      u, & m=n,\\[1ex]
      f_{m-1}^+ \circ \cdots \circ f_n^+(u), & m>n.
    \end{cases}
  \]
\end{definition}

\begin{lemma}[Coherence of the canonical stage embedding]\label{lem:stage-embedding}
  The thread of Definition~\ref{def:stage-embedding} is coherent.  Moreover
  \[
    \pi_n\circ f_{n,\infty}=\mathrm{id}_{K_n},
    \qquad
    f_{n+1,\infty}\circ f_n^+ = f_{n,\infty}.
  \]
\end{lemma}

\begin{proof}
  For $m<n$, the defining coordinate formulas satisfy
  \[
    f_m^-\bigl(\pi_{m+1}(f_{n,\infty}(u))\bigr)
      = f_m^-\bigl(f_{m+1}^- \circ \cdots \circ f_{n-1}^-(u)\bigr)
      = \pi_m(f_{n,\infty}(u)),
  \]
  while for $m\ge n$ the coherence is exactly the projection-pair equation
  $f_m^-f_m^+=\mathrm{id}$.  Thus the coordinates form a coherent thread.
  The identity $\pi_n\circ f_{n,\infty}=\mathrm{id}_{K_n}$ is immediate from the
  middle clause of the definition, and the compatibility
  $f_{n+1,\infty}\circ f_n^+ = f_{n,\infty}$ follows by comparing the two sides
  coordinatewise.
\end{proof}

\begin{remark}[Thread extensionality]\label{rem:thread-extensionality}
Thread equality in $\Kinf$ is coordinatewise: two elements of $\Kinf$ are equal
iff all of their projections $\pi_n$ agree.  We use this extensionality
repeatedly below.  By contrast, the relation $\simeq$ records only mutual
approximation in the homotopy order and is used to state coherence and density
before exact equalities are extracted.
\end{remark}

\begin{remark}[The canonical higher tower on $\Kinf$]\label{rem:kinfty-canonical-tower}
The higher-cell tower used for $\Kinf$ in Section~\ref{sec:witness} is the
canonical identity-type $\omega$-groupoid on the carrier of $\Kinf$: a $1$-cell
from $x$ to $y$ is an equality $x=y$, a $2$-cell is an equality between such
equalities, and higher cells are iterated equalities.  This choice adds no
extra $\Kinf$-specific higher-conversion axioms and is the specific tower
studied in the present paper.  Consequently, once two points of $\Kinf$ are
distinct, they admit no $1$-cell and hence no higher cell in this particular
tower.  Theorem~\ref{thm:C} should therefore be read as a fixed-span
persistence statement for the canonical tower, not as a classification of all
possible higher structures that one might place on $\Kinf$.
\end{remark}

The domain-theoretic content of the construction is the following.

\begin{lemma}[Finite-stage compact factorization]\label{lem:compact-factor}
  Let $x\in\Kinf$ and let $c\ll x$ be compact-below.  Then there exist
  $n\ge 0$ and $b\ll \pi_n(x)$ in $K_n$ such that
  \[
    c\le f_{n,\infty}(b)\le f_{n,\infty}(\pi_n(x)).
  \]
\end{lemma}

\begin{proof}
  Compact-below information in the inverse limit is detected at a finite stage:
  because $c$ is compact, only finitely many coordinates are needed to witness
  the inequality $c\le x$.  Choose $n$ larger than that finite support and let
  $b:=\pi_n(c)$.  Then $b\ll \pi_n(x)$ in the algebraic stage $K_n$, and the
  coordinate formula of Definition~\ref{def:stage-embedding} shows that the
  embedded stage-$n$ thread $f_{n,\infty}(b)$ dominates $c$ on every
  coordinate while still lying below the stage-$n$ approximant
  $f_{n,\infty}(\pi_n(x))$.
\end{proof}

\begin{lemma}[Coordinatewise least upper bounds are coherent]\label{lem:coordinate-lub}
  Let $X\subseteq \Kinf$ be upper-bounded and set
  \[
    u_n:=\sup \pi_n(X)\in K_n.
  \]
  Then
  \[
    f_n^-(u_{n+1})=u_n
    \qquad\text{for all }n\ge 0.
  \]
  Consequently $(u_n)_{n\ge 0}$ is a coherent thread in $\Kinf$ and is the
  least upper bound of $X$.
\end{lemma}

\begin{proof}
  For every $x\in X$, thread coherence gives
  \[
    f_n^-(\pi_{n+1}(x))=\pi_n(x).
  \]
  Since $f_n^-$ is monotone and preserves least upper bounds of upper-bounded
  families, applying it to $u_{n+1}=\sup \pi_{n+1}(X)$ yields an upper bound of
  $\pi_n(X)$.  Conversely, every upper bound of $\pi_n(X)$ pulls back along
  $f_n^-$ to an upper bound of $\pi_{n+1}(X)$ by the projection-pair equations,
  so leastness forces $f_n^-(u_{n+1})=u_n$.  The thread condition follows, and
  coordinatewise leastness then shows that $(u_n)_n$ is $\sup X$ in $\Kinf$.
\end{proof}

\begin{proposition}[Proposition~4.1]
  $\Kinf$ is a homotopy Scott domain: it is algebraic and bounded complete.
\end{proposition}

\begin{proof}
  We first analyze the finite stages.

  \smallskip
  \noindent\emph{Bounded completeness of the stages.}
  The base $K_0=N^+$ is a flat domain, hence bounded complete.
  If $K_n$ is bounded complete, then the function space
  $K_{n+1}=[K_n\to K_n]$ is bounded complete under the pointwise order, since
  least upper bounds of upper-bounded families are computed pointwise.
  Therefore every $K_n$ is bounded complete.

  \smallskip
  \noindent\emph{Algebraicity of the stages.}
  The base stage $K_0$ is algebraic: every non-bottom element is compact, and
  every element is the least upper bound of the compact elements below it.
  Assume now that $K_n$ is algebraic. We show that
  $K_{n+1}=[K_n\to K_n]$ is algebraic.

  For $a,b\in K_n$ with $a$ compact and $b$ compact-below $f(a)$, let
  $[a\Rightarrow b]$ denote the usual compact step map which sends arguments
  above $a$ to $b$ and arguments below $a$ to bottom. Finite joins of such step
  maps are still compact. Let $\mathcal{A}_f$ be the set of all finite joins of
  compact step maps lying below $f$. Then $\mathcal{A}_f$ is directed under
  inclusion of the underlying finite families, and every member of
  $\mathcal{A}_f$ is compact-below $f$.

  To see that $f$ is the least upper bound of $\mathcal{A}_f$, fix
  $x\in K_n$. Since $K_n$ is algebraic, the value $f(x)$ is the least upper
  bound of the compact elements below it. Each such compact element is already
  captured by some step approximation built from a compact approximation to $x$.
  Hence evaluation at $x$ of the directed family $\mathcal{A}_f$ is cofinal in
  the compact-below approximants to $f(x)$, so its pointwise supremum is exactly
  $f(x)$. Since this holds for every $x$, $f=\sup \mathcal{A}_f$, and $K_{n+1}$
  is algebraic.

  Thus every finite stage $K_n$ is an algebraic bounded-complete domain.

  \smallskip
  \noindent\emph{Passage to the inverse limit.}
  For $x\in\Kinf$, define its $n$th finite-stage approximation by
  \[
    a_n(x):=f_{n,\infty}(\pi_n(x)),
  \]
  where $f_{n,\infty}:K_n\to\Kinf$ is the canonical exact embedding of the $n$th
  stage into the limit. By construction,
  \[
    \pi_m(a_n(x))=\pi_m(x)\qquad (m\le n),
  \]
  so each $a_n(x)$ agrees with $x$ on the first $n+1$ coordinates; in
  particular $a_n(x)\le x$. Moreover the family $(a_n(x))_{n\ge 0}$ is directed,
  since passing from stage $n$ to stage $n+1$ only refines the approximation.

  Now let $c\ll x$ be compact-below in $\Kinf$.  By
  Lemma~\ref{lem:compact-factor}, there is some stage $n$ and some compact
  $b\ll \pi_n(x)$ in $K_n$ such that
  \[
    c\le f_{n,\infty}(b)\le a_n(x).
  \]
  Hence the directed family $\{a_n(x)\}_{n\ge 0}$ is cofinal among compact
  approximants below $x$. Therefore $x$ is a least upper bound of this family,
  and the chosen supremum of the family is equivalent to $x$ in the induced
  homotopy order:
  \[
    x \simeq \sup_{n\ge 0} a_n(x).
  \]
  Thus every element of $\Kinf$ is a least upper bound of compact elements below
  it. This proves algebraicity of $\Kinf$.

  Finally, bounded completeness of $\Kinf$ is obtained coordinatewise.  If
  $X\subseteq \Kinf$ is upper-bounded, then each projected family
  $\pi_n(X)\subseteq K_n$ is upper-bounded and therefore has a least upper bound
  in $K_n$.  Lemma~\ref{lem:coordinate-lub} shows that these coordinatewise
  suprema again form a coherent thread, which is the least upper bound of $X$ in
  $\Kinf$.
\end{proof}

\subsection{Exact reflexive packaging}\label{sec:packaging}

The central structural fact is that the inverse limit is itself reflexive.  In
fact, the projection-pair machinery on the tower produces global reify and
reflect maps.  The proof of this fact assembles three technical ingredients that
are implicit in the inverse-limit construction but are worth isolating as
standalone statements, both because they clarify the argument and because they
are reused in the proof-relevant witness analysis of
Section~\ref{sec:witness}.

\begin{lemma}[Density of finite-stage approximants]\label{lem:density}
  For every $x\in\Kinf$, define the \emph{$n$th finite-stage approximant}
  \[
    a_n(x) := f_{n,\infty}(\pi_n(x)) \in \Kinf.
  \]
  Then:
  \begin{enumerate}[label=(\roman*)]
    \item $a_n(x)\le x$ for every $n$,
    \item the family $(a_n(x))_{n\ge 0}$ is directed, and
    \item $x \simeq \sup_{n\ge 0} a_n(x)$.
  \end{enumerate}
  In particular, $\Kinf$ is generated, up to chosen-supremum equivalence, by
  the embedded finite stages $\bigcup_n f_{n,\infty}(K_n)$.
\end{lemma}

\begin{proof}
  Clause~(i) holds because $a_n(x)$ agrees with $x$ on coordinates $0$
  through $n$ (by the retract property $\pi_n\circ f_{n,\infty}=\mathrm{id}$)
  and is $\le x$ on all higher coordinates (by iterated application of the
  section inequality $f_m^+f_m^-\le\mathrm{id}$).  Clause~(ii) follows because
  passing from stage $n$ to $n+1$ only refines the approximation: $a_n(x)\le
  a_{n+1}(x)$ by the compatibility of the projection pairs.  Clause~(iii) was
  established in the proof of Proposition~4.1 (algebraicity of $\Kinf$): every
  compact element below $x$ is already below some $a_n(x)$, so the directed
  family is cofinal among the compact approximants to $x$, hence has $x$ as a
  least upper bound; the chosen supremum is therefore equivalent to $x$.
\end{proof}

\begin{lemma}[Coherence of stage restrictions]\label{lem:coherent-restrict}
  Let $g:\Kinf\to\Kinf$ be a continuous endomap and define its \emph{stage-$n$
  restriction} by
  \[
    r_n(g) := \pi_n\circ g\circ f_{n,\infty} \;\in\; [K_n\to K_n] = K_{n+1}.
  \]
  Then the family $(r_n(g))_{n\ge 0}$ satisfies the coherence condition
  \[
    f_n^-\bigl(r_{n+1}(g)\bigr) = r_n(g)
    \qquad\text{for all } n\ge 0.
  \]
  In particular, the shifted thread $\bigl(f_0^-(r_0(g)),\,r_0(g),\,r_1(g),\,
  \ldots\bigr)$ is a well-defined element of $\Kinf$.
\end{lemma}

\begin{proof}
  Fix $a\in K_n$.  Applying $f_n^-$ to the map $r_{n+1}(g)$ and evaluating at
  $a$ gives
  \[
    f_n^-\bigl(r_{n+1}(g)\bigr)(a)
    = \pi_n\bigl(g(f_{n+1,\infty}(f_n^+(a)))\bigr).
  \]
  Since $f_{n+1,\infty}\circ f_n^+$ and $f_{n,\infty}$ produce the same
  coherent thread in $\Kinf$ (both agree on all coordinates up to $n$ and
  satisfy the tower coherence above $n$), continuity of $g$ gives
  $g(f_{n+1,\infty}(f_n^+(a))) = g(f_{n,\infty}(a))$.  Projecting to level $n$
  yields $r_n(g)(a)$.  Since $a$ was arbitrary, $f_n^-(r_{n+1}(g))=r_n(g)$ as
  required.  The ``in particular'' clause follows because this coherence
  condition is exactly the defining property of a thread in the inverse limit
  $\Kinf$.
\end{proof}

\begin{lemma}[Monotonicity of application shadows]\label{lem:application-shadows}
  For $x,y\in\Kinf$, define
  \[
    k_n(x,y):=f_{n,\infty}\bigl(\pi_{n+1}(x)(\pi_n(y))\bigr)\in\Kinf.
  \]
  Then
  \[
    k_n(x,y)\le k_{n+1}(x,y)
    \qquad\text{for all }n\ge 0.
  \]
  Consequently the directed supremum
  \[
    k(x,y):=\sup_{n\ge 0}k_n(x,y)
  \]
  exists in $\Kinf$.
\end{lemma}

\begin{proof}
  By thread coherence of $x$, the coordinate $\pi_n(x)$ is the
  $f_n^-$-projection of $\pi_{n+1}(x)$.  Applying the section inequality
  $f_n^+f_n^-\le \mathrm{id}$ at stage $n+1$ shows that evaluation at the lifted
  argument in stage $n+1$ dominates the stage-$n$ value.  Embedding those
  values back into $\Kinf$ via the compatible maps $f_{n,\infty}$ and
  $f_{n+1,\infty}$ yields $k_n(x,y)\le k_{n+1}(x,y)$.
\end{proof}

\begin{lemma}[Agreement on embedded finite stages]\label{lem:stage-agreement}
  Let $f,g:\Kinf\to\Kinf$ be continuous endomaps.  Suppose that for every
  $n\ge 0$ and every $y\in K_n$ one has
  \[
    f(f_{n,\infty}(y))=g(f_{n,\infty}(y)).
  \]
  Then $f=g$.
\end{lemma}

\begin{proof}
  Fix $z\in\Kinf$.  By Lemma~\ref{lem:density}, the finite-stage approximants
  $a_n(z)=f_{n,\infty}(\pi_n(z))$ form a directed family whose chosen supremum
  is equivalent to $z$.  By hypothesis, $f$ and $g$ agree on every $a_n(z)$.
  Since both maps are continuous, they agree on the chosen supremum of that
  family.  Thread equality in $\Kinf$ is coordinatewise
  (Remark~\ref{rem:thread-extensionality}), and antisymmetry at each finite
  stage turns the chosen-supremum equivalence back into literal equality with
  $z$.  Hence $f(z)=g(z)$.  Since $z$ was arbitrary, $f=g$.
\end{proof}

\begin{lemma}[Coordinatewise continuity into the inverse limit]\label{lem:coordinatewise-cont}
  Let $A$ be a complete homotopy partial order and let
  $g_n:A\to K_n$ be continuous maps satisfying
  \[
    f_n^-(g_{n+1}(a))=g_n(a)
    \qquad\text{for every }a\in A\text{ and every }n\ge 0.
  \]
  Then the map
  \[
    a\longmapsto (g_n(a))_{n\ge 0}
  \]
  defines a continuous function $A\to\Kinf$.
\end{lemma}

\begin{proof}
  The tuple $a\mapsto (g_n(a))_n$ is continuous coordinatewise by hypothesis.
  The displayed coherence equations show that its image lands in the inverse
  limit.  Since the order and topology on $\Kinf$ are coordinatewise, this
  coherent tuple defines a continuous map into $\Kinf$.
\end{proof}

The following theorem is therefore reduced to five explicit ingredients:
Lemma~\ref{lem:density} for density of finite-stage approximants,
Lemma~\ref{lem:coherent-restrict} for well-definedness of the reify map,
Lemma~\ref{lem:application-shadows} for the directed application shadow, and
Lemma~\ref{lem:stage-agreement} for equality of continuous maps from dense
finite-stage agreement, together with
Lemma~\ref{lem:coordinatewise-cont} for continuity of coherent coordinate
families into the inverse limit.

\begin{theorem}[Exact $\Kinf$ reflexive packaging]
\label{thm:D}
  There exist globally continuous maps
  \[
    h : [\Kinf \to \Kinf] \to \Kinf,
    \qquad
    k : \Kinf \to [\Kinf \to \Kinf]
  \]
  such that
  \[
    k \circ h = \mathrm{id}_{[\Kinf \to \Kinf]},
    \qquad
    h \circ k = \mathrm{id}_{\Kinf}.
  \]
  The induced application operation
  \[
    \mathsf{app}(x,y):=k(x)(y)
  \]
  is jointly continuous. If $f_{n,\infty}:K_n\to\Kinf$ denotes the canonical
  stage embedding, then for every $x\in\Kinf$, every $y\in K_n$, and every
  $n\ge 0$,
  \[
    \pi_n\bigl(\mathsf{app}(x,f_{n,\infty}(y))\bigr)
      = \pi_{n+1}(x)(y).
  \]
\end{theorem}

\begin{proof}
  We construct $h$ and $k$ directly from the finite stages.

  \smallskip
  \noindent\emph{Step 1: exact finite-stage embeddings.}
  Definition~\ref{def:stage-embedding} and
  Lemma~\ref{lem:stage-embedding} provide the canonical embeddings
  \[
    f_{n,\infty}:K_n\to\Kinf
  \]
  together with the exact identities
  \[
    \pi_n\circ f_{n,\infty}=\mathrm{id}_{K_n},
    \qquad
    f_{n+1,\infty}\circ f_n^+ = f_{n,\infty}.
  \]

  \smallskip
  \noindent\emph{Step 2: finite-stage restrictions of endomaps.}
  Given a continuous endomap $g:\Kinf\to\Kinf$, define its stage-$n$ restriction
  by
  \[
    r_n(g):=\pi_n\circ g\circ f_{n,\infty} \in [K_n\to K_n]=K_{n+1}.
  \]
  These restrictions are compatible:
  \[
    f_n^-\bigl(r_{n+1}(g)\bigr)=r_n(g).
  \]
  Indeed, for $a\in K_n$,
  \[
    f_n^-\bigl(r_{n+1}(g)\bigr)(a)
      = \pi_n\bigl(g(f_{n+1,\infty}(f_n^+(a)))\bigr)
      = \pi_n\bigl(g(f_{n,\infty}(a))\bigr)
      = r_n(g)(a),
  \]
  because $f_{n+1,\infty}\circ f_n^+$ and $f_{n,\infty}$ define the same thread
  in the inverse limit.

  The coherence condition of Lemma~\ref{lem:coherent-restrict} guarantees
  that we may assemble the coherent thread
  \[
    h(g)_0:=f_0^-(r_0(g)),\qquad h(g)_{n+1}:=r_n(g)\ \ (n\ge 0).
  \]
  Each coordinate $g\mapsto h(g)_n$ is continuous because it is built by
  composing continuous maps, and Lemma~\ref{lem:coherent-restrict} supplies the
  required thread equations.  Lemma~\ref{lem:coordinatewise-cont} therefore
  yields a continuous map
  \[
    h:[\Kinf\to\Kinf]\to\Kinf.
  \]

  \smallskip
  \noindent\emph{Step 3: finite-stage application shadows.}
  For $x,y\in\Kinf$, define the $n$th binary approximation to application by
  \[
    k_n(x,y):=f_{n,\infty}\bigl(\pi_{n+1}(x)(\pi_n(y))\bigr)\in\Kinf.
  \]
  Each $k_n$ is jointly continuous, and Lemma~\ref{lem:application-shadows}
  shows that the family $(k_n)_{n\ge 0}$ is monotone.  Hence the directed
  supremum
  \[
    k(x,y):=\sup_{n\ge 0}k_n(x,y)
  \]
  exists in $\Kinf$.  For each fixed coordinate $m$, the map
  \[
    (x,y)\longmapsto \pi_m(k(x,y))
  \]
  is the directed supremum of the continuous coordinate maps
  $(x,y)\mapsto \pi_m(k_n(x,y))$, hence is continuous.  The stage-embedding
  identities of Lemma~\ref{lem:stage-embedding} and the monotonicity from
  Lemma~\ref{lem:application-shadows} ensure that these coordinates satisfy the
  inverse-limit coherence equations.  Lemma~\ref{lem:coordinatewise-cont}
  therefore gives a jointly continuous map
  \[
    (x,y)\longmapsto k(x,y)\in\Kinf,
  \]
  and currying gives
  \[
    k:\Kinf\to[\Kinf\to\Kinf].
  \]

  \smallskip
  \noindent\emph{Step 4: the stagewise formula.}
  Fix $x\in\Kinf$ and $y\in K_n$. Restrict the map $k(x)$ along the exact
  embedding $f_{n,\infty}$. On the one hand, the $n$th approximant
  $k_n(x,f_{n,\infty}(y))$ has $n$th coordinate exactly
  $\pi_{n+1}(x)(y)$. On the other hand, every higher approximant has $n$th
  coordinate below the same value, so after projecting to level $n$ the entire
  directed family is bounded above by $\pi_{n+1}(x)(y)$ and contains an element
  attaining that upper bound. Since projection preserves directed suprema,
  \[
    \pi_n\bigl(k(x)(f_{n,\infty}(y))\bigr)=\pi_{n+1}(x)(y).
  \]
  This is exactly the stated stagewise application formula.

  \smallskip
  \noindent\emph{Step 5: the section law $k\circ h=\mathrm{id}$.}
  Let $g:\Kinf\to\Kinf$. For every stage $n$ and every $y\in K_n$, the previous
  formula applied to $h(g)$ gives
  \[
    \pi_n\bigl(k(h(g))(f_{n,\infty}(y))\bigr)
      = \pi_{n+1}(h(g))(y)
      = r_n(g)(y)
      = \pi_n\bigl(g(f_{n,\infty}(y))\bigr).
  \]
  Thus $k(h(g))$ and $g$ agree on every embedded finite stage, and
  Lemma~\ref{lem:stage-agreement} implies
  \[
    k(h(g))=g.
  \]

  \smallskip
  \noindent\emph{Step 6: the retract law $h\circ k=\mathrm{id}_{\Kinf}$.}
  For $x\in\Kinf$ and $n\ge 0$,
  \[
    \pi_{n+1}(h(k(x)))=r_n(k(x)).
  \]
  Evaluating this map at $y\in K_n$ and using the stagewise formula just proved,
  we obtain
  \[
    \pi_{n+1}(h(k(x)))(y)
      = \pi_n\bigl(k(x)(f_{n,\infty}(y))\bigr)
      = \pi_{n+1}(x)(y).
  \]
  Therefore the $(n+1)$st coordinates of $h(k(x))$ and $x$ coincide for all
  $n\ge 0$. At level $0$ one gets
  \[
    \pi_0(h(k(x)))=f_0^-(\pi_1(x))=\pi_0(x)
  \]
  by thread coherence. Hence, by Remark~\ref{rem:thread-extensionality},
  \[
    h(k(x))=x.
  \]

  Finally, $\mathsf{app}(x,y)=k(x)(y)$ is jointly continuous because it is the
  composite of the continuous map
  \[
    (x,y)\longmapsto (k(x),y)
  \]
  with the continuous evaluation map on the function space
  $[\Kinf\to\Kinf]$.
\end{proof}

\begin{remark}[Relation to classical projection-pair models]
Classical projection-pair inverse-limit constructions already motivate
reflexive objects such as $D_\infty$~\cite{Scott:continuous,AbramskyJung:domain}.
The two-sided isomorphism $h\circ k=\mathrm{id}$ and
$k\circ h=\mathrm{id}$ is therefore not claimed here as a new abstract
inverse-limit theorem.  The specific new content of Theorem~\ref{thm:D} is the
stagewise evaluation formula
\[
\pi_n\bigl(\mathsf{app}(x,f_{n,\infty}(y))\bigr)=\pi_{n+1}(x)(y),
\]
which characterizes application at every finite stage and is used in
Section~\ref{sec:witness} to pin down the witness interpretations at the level
of $K_0$ coordinates.  The explicit formulas for $h$ and $k$ given in the
proof also provide a constructive, self-contained verification of this concrete
$\Kinf$ instance rather than an appeal to the general inverse-limit theorem
alone.
\end{remark}

The maps $h$ and $k$ are implemented in the Lean formalization by projection-pair
data at each finite stage.  The following code fragment illustrates that
interface; Section~\ref{sec:witness} then uses the stagewise formula of
Theorem~\ref{thm:D} to separate proof-relevant witness classes inside $\Kinf$.

\begin{lstlisting}[caption={Projection pairs for the $K_n$ tower
  (from \texttt{KInfinity/Construction.lean}).}]
structure ProjectionPair
    (A : CompleteHomotopyPartialOrder)
    (B : CompleteHomotopyPartialOrder) where
  emb : ContinuousMap A B
  proj : ContinuousMap B A
  retract : ContinuousMap.comp proj emb = ContinuousMap.id A
  section_le :
    forall x : B.Obj, B.Rel (emb (proj x)) x
\end{lstlisting}

\section{Fixed-Span Witness Classification and Separation}\label{sec:witness}

Paper~II~\cite{MRdQ:Kinfinity} established, via the witness-separation span
restated below as Example~\ref{ex:witness-span}, that the $\Kinf$
model is nontrivial: a distinguished $\beta/\eta$ witness pair from the same
redex receives semantically distinct images.  What it did not do was analyze
the fixed forward witness language on that span or trace, in the canonical
higher tower on $\Kinf$, how the resulting point-level separation propagates
through all higher dimensions.  The present section remains deliberately local:
it is a theorem about this fixed span, not a global witness calculus for
arbitrary source/target pairs.

The present section has two logically distinct parts.  First, we introduce a
\emph{witness language} that remembers which contraction was chosen and, on the
fixed witness-separation span, classifies every witness in the chosen forward
fragment by a canonical $\beta$ or $\eta$ tag.  This is the model-specific
input.  Second, we show that once the
resulting point-level dichotomy is established, the absence of connecting
$1$-cells and higher cells is forced by the equality-generated structure of the
$\Kinf$ higher-cell tower, including the recursively completed dimensions.

\subsection{A basic witness-separation example}

Paper~II~\cite{MRdQ:Kinfinity} proved the non-triviality of $\Kinf$ by means of
the following concrete span, which we restate here for direct reference.

\begin{example}[Witness-separation span]\label{ex:witness-span}
Consider the source
\[
  M=((\lambda z.\,xz)\,y)
\]
and the common target
\[
  N=xy.
\]
The two basic one-step contractions are displayed explicitly as
\[
  ((\lambda z.\,xz)\,y)\;\to_\beta\;xz[y/z]=xy
\]
and
\[
  ((\lambda z.\,xz)\,y)\;\to_\eta\;xy,
\]
where in the second line the $\eta$-contraction is performed in the function
part $\lambda z.\,xz\to_\eta x$.  Thus the same source and target admit two
distinguished one-step witnesses: the direct $\beta$-contraction and the direct
$\eta$-contraction through the function part.
\end{example}

\begin{remark}[Notation for the running span]
Although the formal development works uniformly with de~Bruijn indices, we
display Example~\ref{ex:witness-span} in named-variable notation for
readability.  The free variables $x$ and $y$ should therefore be read only as
mnemonic placeholders for a fixed open term; all subsequent classification and
semantic arguments remain statements about the de~Bruijn formalism fixed in
Section~\ref{sec:prelim}.
\end{remark}

The point of the proof-relevant refinement is that we do not merely remember
that the endpoints are propositionally $\beta\eta$-equal; we remember
\emph{which contraction witness was chosen}.  A general proof-relevant explicit
$1$-witness from $M$ to $N$ may contain reflexive padding before or after the
essential contraction, but within this witness-separation span the witness
language still falls into exactly two reachability classes: the $\beta$-class
and the $\eta$-class.  Thus every proof-relevant
$1$-witness on this span carries a canonical tag
\[
  \mathrm{tag}(t)\in\{\beta,\eta\}.
\]
Figure~\ref{fig:witness-separation} records this fixed span and the two
canonical semantic images attached to its witness classes.

\begin{figure}[htbp]
\centering
\begin{tikzpicture}[
  term/.style={draw, rounded corners, align=center, inner sep=6pt, fill=gray!8},
  semantic/.style={draw, rounded corners, align=center, text width=4.2cm, inner sep=5pt, fill=gray!4},
  witness/.style={-{Latex[length=2mm]}, thick},
  sem/.style={-{Latex[length=2mm]}, thick, dashed}
]
  \node[term] (M) {$M=((\lambda z.\,xz)\,y)$};
  \node[term, right=5.5cm of M] (N) {$N=xy$};

  \draw[witness] (M) to[bend left=22] node[above] {$t_\beta$} (N);
  \draw[witness] (M) to[bend right=22] node[below] {$t_\eta$} (N);

  \path (M) -- coordinate[midway] (center) (N);
  \path (M) to[bend left=22] coordinate[midway] (betamid) (N);
  \path (M) to[bend right=22] coordinate[midway] (etamid) (N);

  \node[semantic, below=2.35cm of center, xshift=-2.55cm] (R)
    {$\beta$-class\\[2pt]{\footnotesize $\llbracket t \rrbracket_{\Kinf}$}\\[-1pt]{\footnotesize $= f_{0,\infty}(s^R_1)$}};
  \node[semantic, below=2.35cm of center, xshift=2.55cm] (L)
    {$\eta$-class\\[2pt]{\footnotesize $\llbracket t \rrbracket_{\Kinf}$}\\[-1pt]{\footnotesize $= f_{0,\infty}(s^L_1)$}};

  \draw[sem] (betamid) -- (R.north);
  \draw[sem] (etamid) -- (L.north);

  \node at ($(R.east)!0.5!(L.west)+(0,0.5)$) {$\neq$};
\end{tikzpicture}
\caption{The witness-separation span of Example~\ref{ex:witness-span}.  The same source and target admit a $\beta$-witness and an $\eta$-witness, but their canonical $\Kinf$ interpretations land at distinct points.}
\label{fig:witness-separation}
\end{figure}
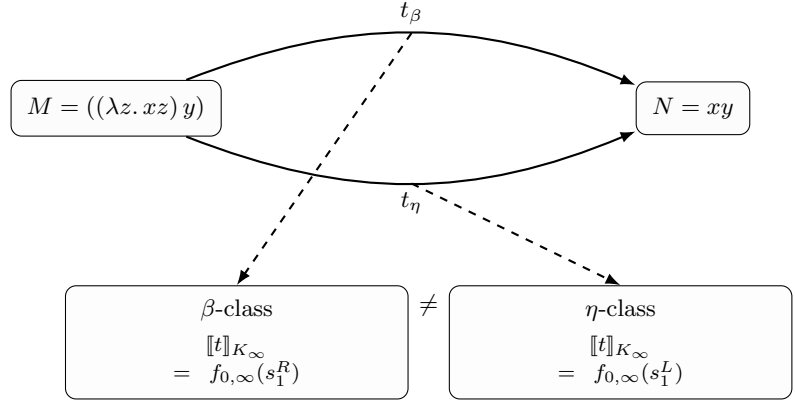

\subsection{Witness languages}

\begin{definition}[Fixed-span $1$-term witnesses]
  Fix the source $M$ and target $N$ of Example~\ref{ex:witness-span}.  Let
  $W(A,B)$, for $A,B\in\{M,N\}$, be the inductively generated typed family of
  witnesses generated by:
  \begin{enumerate}[label=(\arabic*)]
    \item the direct witnesses $t_\beta:W(M,N)$ and $t_\eta:W(M,N)$;
    \item the reflexive witnesses $\mathsf{refl}_M:W(M,M)$ and
      $\mathsf{refl}_N:W(N,N)$;
    \item composition: if $u:W(A,B)$ and $v:W(B,C)$, then
      $u\cdot v:W(A,C)$.
  \end{enumerate}
  A \emph{$1$-term witness} from $M$ to $N$ is an element of $W(M,N)$.  There is
  no inverse constructor in this language.
\end{definition}

\begin{remark}[Scope of the witness language]\label{rem:witness-language-scope}
The language $W(A,B)$ is intentionally much smaller than the full
reduction-sequence language or the full low-dimensional higher-conversion
syntax.  It keeps only the two forward one-step witnesses on the fixed span,
reflexivity, and composition.  In particular, it does not attempt to classify
arbitrary reduction sequences, inverses, whiskering, or congruence data.
Theorem~\ref{thm:C} is therefore a fixed-span statement about this forward
fragment only.
\end{remark}

In this section we restrict attention to the fixed witness-separation span of
Example~\ref{ex:witness-span}
\[
  M=((\lambda z.\,xz)\,y),
  \qquad
  N=xy.
\]
To each explicit $1$-term witness $t$ from $M$ to $N$ we attach two pieces of
semantic data.  This interpretation is intentionally local: it classifies the
fixed-span witness language $W(M,N)$ by the two distinguished points of
$\Kinf$ attached to Example~\ref{ex:witness-span}.

\begin{enumerate}[label=(\arabic*), leftmargin=*, labelsep=0.6em]
  \item a continuous equality $\epsilon_t$ witnessing that the source and target
  terms have the same ordinary denotation in the reflexive semantics, and
  \item a distinguished point of $\Kinf$ recording whether the witness belongs
  to the $\beta$-class or to the $\eta$-class.
\end{enumerate}

The second piece is what detects proof relevance.  Let $s^R_1$ and $s^L_1$ be
the right and left poles in the $S^1$-summand of the non-trivial base object
used to define $K_0=N^+$.  If $\mathrm{tag}(t)=\beta$, we set
\[
  \llbracket t\rrbracket_{\Kinf}:=f_{0,\infty}(s^R_1),
\]
and if $\mathrm{tag}(t)=\eta$, we set
\[
  \llbracket t\rrbracket_{\Kinf}:=f_{0,\infty}(s^L_1).
\]
Thus the interpretation depends only on the essential witness class, not on the
amount of reflexive padding surrounding it.  We now make this classification and
its stability precise before stating the main separation theorem.

\begin{lemma}[Witness classification on the witness-separation span]\label{lem:witness-normal}
  Let $t$ be an explicit $1$-term witness from the source $M =
  ((\lambda z.\,xz)\,y)$ to the target $N = xy$.  Then exactly one of the
  following holds:
  \begin{enumerate}[label=(\alph*)]
    \item $t$ is obtained from the direct $\beta$-contraction by inserting
      reflexive witnesses on the left and/or right;
    \item $t$ is obtained from the direct $\eta$-contraction by inserting
      reflexive witnesses on the left and/or right.
  \end{enumerate}
  In particular, the canonical tag $\mathrm{tag}(t)\in\{\beta,\eta\}$ is
  well-defined.
\end{lemma}

\begin{proof}
  Proceed by structural induction on the explicit witness language.  The base
  constructors are the direct $\beta$- and $\eta$-witnesses together with
  reflexivity.  The reflexive case cannot occur between the chosen source and
  target because $M\neq N$, so any witness from $M$ to $N$ must be classified
  by one of the two nontrivial contractions.

  For a composite witness, once a nontrivial $\beta$- or $\eta$-witness has
  appeared, the remaining factor can contribute only reflexive padding.  Indeed,
  the witness-separation span admits no third nontrivial contraction pattern: the
  source has exactly the direct outer $\beta$-contraction and the inner
  $\eta$-contraction through the function part, and the target $xy$ supports no
  further forward $\beta$- or $\eta$-step.  Hence every sequential composite
  factors as one distinguished nontrivial contraction together with reflexive
  witnesses on the left and/or right.

  The two classes are disjoint because the direct $\beta$- and $\eta$-steps are
  syntactically distinct and have different chosen semantic tags.  Therefore
  exactly one of (a) or (b) holds, and the tag is well-defined.
\end{proof}

\begin{corollary}[Stability under reflexive padding]\label{cor:tag-stable}
  Let $t$ and $t'$ be explicit $1$-term witnesses from $M$ to $N$ that differ
  only by insertion or removal of reflexive subpaths.  Then
  $\mathrm{tag}(t)=\mathrm{tag}(t')$, and hence
  $\llbracket t \rrbracket_{\Kinf} = \llbracket t' \rrbracket_{\Kinf}$.
\end{corollary}

\begin{proof}
  By Lemma~\ref{lem:witness-normal}, the two classes are defined precisely by
  adjoining reflexive witnesses to the direct $\beta$- or $\eta$-contraction.
  Inserting or removing reflexive padding therefore cannot change whether a
  witness is of $\beta$-type or $\eta$-type.  Hence
  $\mathrm{tag}(t)=\mathrm{tag}(t')$, and the $\Kinf$-interpretation was
  defined in terms of the tag alone.
\end{proof}

With these preliminaries in hand, we can state the main fixed-span persistence
theorem.

\begin{theorem}[Fixed-span witness classification and separation in the canonical $\Kinf$ tower]
\label{thm:C}
  Let $t_\beta$ and $t_\eta$ be the explicit $\beta_1$- and $\eta_1$-witnesses
  from Example~\ref{ex:witness-span}.  In the canonical higher tower on
  $\Kinf$ described in Remark~\ref{rem:kinfty-canonical-tower}, their canonical
  interpretations are:
  \begin{enumerate}[label=(\roman*)]
    \item distinct as $\Kinf$-points,
    \item not connected by any $1$-cell, i.e., by any equality in the carrier of
      $\Kinf$,
    \item consequently not connected by any higher cell in any positive
      dimension of that tower.
  \end{enumerate}
  The same conclusion holds for any pair of witnesses in the fixed-span
  language $W(M,N)$ whose canonical tags are $\beta$ and $\eta$,
  respectively.
\end{theorem}

\begin{proof}
  By Lemma~\ref{lem:witness-normal}, every witness in the fixed-span language
  $W(M,N)$ lies in exactly one of the two classes obtained from the direct
  $\beta$- or $\eta$-contraction by reflexive padding, and
  Corollary~\ref{cor:tag-stable} shows that the tag is invariant under
  insertion or removal of those reflexive subpaths.  Thus every witness with
  $\beta$-tag is interpreted by the point $f_{0,\infty}(s^R_1)$, while every
  witness with $\eta$-tag is interpreted by $f_{0,\infty}(s^L_1)$.

  For the distinguished witnesses $t_\beta$ and $t_\eta$, we therefore have
  \[
    \pi_0(\llbracket t_\beta\rrbracket_{\Kinf})=s^R_1,
    \qquad
    \pi_0(\llbracket t_\eta\rrbracket_{\Kinf})=s^L_1.
  \]
  Since $s^R_1$ and $s^L_1$ are distinct poles of the chosen $S^1$-summand,
  the two interpreted points are distinct in $\Kinf$.  This proves~(i).

  By Remark~\ref{rem:kinfty-canonical-tower}, a $1$-cell in the canonical
  $\Kinf$ tower is literally an equality in the carrier of $\Kinf$.  Hence the
  distinct points $\llbracket t_\beta\rrbracket_{\Kinf}$ and
  $\llbracket t_\eta\rrbracket_{\Kinf}$ admit no $1$-cell, proving~(ii).
  Higher cells in that tower are iterated equalities between parallel lower
  cells, so once the boundary $1$-cell type is empty there can be no higher cell
  in any positive dimension with these endpoints.  This proves~(iii).

  The final sentence follows because every witness in $W(M,N)$ with $\beta$-tag
  is interpreted by the same canonical $\beta$-point and every witness with
  $\eta$-tag by the same canonical $\eta$-point.
\end{proof}

\begin{remark}
  Theorem~\ref{thm:C} should be read as a fixed-span persistence statement.
  Its model-specific input is the point-level $\beta/\eta$ distinction on the
  witness language of Example~\ref{ex:witness-span}; the higher-dimensional
  non-connection is then forced by the canonical identity-type tower chosen on
  $\Kinf$.  In particular, the theorem does not claim a global compositional
  witness semantics for arbitrary pairs of terms, nor does it classify
  alternative higher structures that one might place on $\Kinf$.
\end{remark}

\section{Discussion and Future Work}\label{sec:discussion}

\subsection{Summary of results}

We have presented four theorem packages on the theory of higher
$\lambda$-models:
\begin{enumerate}[label=(\Alph*)]
  \item Structural comparison with the recursive completion for higher
    $\lambda$-conversions (Theorem~\ref{thm:A}).
  \item A smaller sufficient semantic interface for associator/pentagon reasoning
    (Theorem~\ref{thm:B}).
  \item Fixed-span witness classification and separation in the canonical
    $\Kinf$ tower
    (Theorem~\ref{thm:C}).
  \item Exact reflexive packaging and application for $\Kinf$
    (Theorem~\ref{thm:D}).
\end{enumerate}

\noindent
The four theorems do not all have the same rôle.  Theorems~\ref{thm:B}
and~\ref{thm:D} contain the main generic-semantic and concrete-model input.
Theorem~\ref{thm:A} records the finite packaging boundary between the explicit
syntax and its recursive continuation, while Theorem~\ref{thm:C} explains what
the fixed-span point-level separation implies inside the canonical higher tower
on $\Kinf$.  The net claim of the paper is therefore not four co-equal
breakthroughs, but a division of labour between a small semantic input and the
structural consequences forced by the chosen towers.

\subsection{Comparison with related work}

We organize the discussion into three threads: direct antecedents within the
higher $\lambda$-model programme, the broader higher-categorical and
homotopical background, and the classical domain-theoretic lineage.

\medskip
\noindent\textbf{Direct antecedents.}
The immediate foundation is the two-paper programme of Mart\'{\i}nez-Rivillas
and de~Queiroz.  Paper~I~\cite{MRdQ:theory} introduced the extensional Kan
semantics, defined the explicit tower of $n$-conversions, and proved the main
inclusion $\THeq\subseteq\HoTFT$.  Paper~II~\cite{MRdQ:Kinfinity} constructed
$\Kinf$ and already isolated the distinguished $\beta/\eta$ witness pair
restated here as Example~\ref{ex:witness-span}.  The present work begins from
those results but addresses four specific remaining gaps.  For
Theorem~\ref{thm:A}, the missing point is not the existence of recursive
higher-derivation syntax, but the comparison between the explicit
low-dimensional tower and the equality-generated continuation together with the
proof of strict globular compatibility for that continuation.  For
Theorem~\ref{thm:B}, the missing point is the isolation of the smaller
WLWR-plus-front-seed interface sufficient for the later semantic arguments.
For Theorem~\ref{thm:C}, Paper~II gives only point-level separation; what was
missing was a witness-level classification on the witness-separation span together
with an account of how that separation behaves in the canonical higher tower on
$\Kinf$.  For
Theorem~\ref{thm:D}, the missing point is the exact reflexive packaging of
$\Kinf$ by globally continuous reify, reflect, and application maps with
explicit stagewise formulas and exact coordinate identities.

Earlier in the same programme, Mart\'{\i}nez-Rivillas and
de~Queiroz~\cite{MRdQ:homotopydomain} developed a homotopy domain theory that
enriches the classical Scott semantics with an explicit homotopical layer, and
in~\cite{MRdQ:topological} they studied the $\infty$-groupoid generated by an
arbitrary topological $\lambda$-model.  Both papers anticipate the idea that
  the conversion structure of the $\lambda$-calculus has a nontrivial
  higher-categorical life, but neither constructs the realization comparison of
  Theorem~\ref{thm:A} or the fixed-span witness classification of
  Theorem~\ref{thm:C}.

A closely related syntactic thread is the computational-paths programme of
de~Veras, Ramos, de~Queiroz, and de~Oliveira~\cite{deVeras:computational},
which showed that the paths generated by directed reduction steps in a rewriting
system form a weak groupoid.  Their construction provides a syntactic template
for the 1-dimensional layer of our tower: the reduction sequences between
$\lambda$-terms are precisely computational paths, and the explicit 2-cells
(homotopies between reduction sequences) can be viewed as the coherence data
that upgrade the resulting path groupoid to a $2$-groupoid.  The present paper
pushes that picture beyond the path-groupoid layer by comparing the recursive
completion of the low-dimensional core with the explicit higher-conversion tower
(Theorem~\ref{thm:A}).

\medskip
\noindent\textbf{Higher-categorical and homotopical background.}\quad
The broader higher-categorical references are conceptual rather than technical
antecedents.  Batanin~\cite{Batanin:monoidal} provides the globular language
relevant to our use of source/target compatibility, whiskering, and coherence,
but we do not import his machinery directly.  Likewise,
Grothendieck~\cite{Grothendieck:pursuing}, homotopy type theory~\cite{HoTT},
and Lurie~\cite{Lurie:htt} are cited only as background reminders that higher
identification data can carry nontrivial coherent structure.  The present paper
does not derive results from those frameworks; it studies a concrete
syntactic/semantic construction for the untyped $\lambda$-calculus.  We intend
these comparisons as orientation, not as claims of direct technical overlap.

\medskip
\noindent\textbf{Classical domain-theoretic lineage.}\quad
The $\Kinf$ construction (Section~\ref{sec:Kinfty}) is a direct descendant of
Scott's continuous lattice semantics~\cite{Scott:continuous} and the later
projection-pair account surveyed by Abramsky and
Jung~\cite{AbramskyJung:domain}.  Scott showed that every countably based
continuous lattice admits a reflexive structure, and the resulting models
interpret the untyped $\lambda$-calculus.  The $\Kinf$ tower replaces Scott's
single-step reflexivity with an iterated function-space tower
$K_{n+1}=[K_n\to K_n]$ whose inverse limit inherits both the classical
algebraic-domain structure and a new homotopical layer.

Barendregt's~\cite{Barendregt:lambda} systematic treatment of $\lambda$-models
provides the classical reference frame: the notions of extensionality, the
$\beta\eta$-theory, Church--Rosser, and confluence that appear in our
preliminaries and baseline sections are standard in the sense
of~\cite{Barendregt:lambda}.  The novelty of the present approach is that these
classical ingredients are enriched with higher-dimensional semantic structure
that was invisible in the purely propositional setting.

\medskip
\noindent\textbf{Nature of the contribution.}\quad
The rôle of the formal development is methodological and certifying rather than
expository.  The mathematical burden of the paper is carried by the arguments in
Sections~\ref{sec:all-dim}--\ref{sec:witness}.  Among those, Theorem~\ref{thm:D}
is the strongest model-specific statement, Theorem~\ref{thm:B} identifies a
smaller coherence input currently proved sufficient for the later semantic
comparison arguments, and
Theorems~\ref{thm:A} and~\ref{thm:C} describe structural consequences of the
recursive higher tower once the low-dimensional and point-level data are fixed.
Lean~4~\cite{Lean4} provided a disciplined setting in which these arguments were
discovered and checked, but the paper is intended to stand or fall on its
mathematics rather than on the existence of the formalization.
In the broader logic-and-computation tradition, the results illustrate a
recurring theme: that the proof-relevant refinement of a classical equivalence
relation (here, $\beta\eta$-convertibility) reveals computational
structure---witness classes, coherence data, semantic separation---that the
propositional collapse erases.  The constructive character of the proofs,
particularly the explicit witness constructions, aligns with the programme of
making the computational content of logical derivations visible and tractable.
In short: higher $\lambda$-models can now be read as a small coherence seed, an
exact concrete inverse-limit model, and the recursive higher consequences
forced by those data.

\subsection{Open problems}

\begin{enumerate}[label=\arabic*., leftmargin=*, labelsep=0.6em]
  \item \textbf{Eliminating the front-seed.}
    Derive the front-seed coherence package (seeds FS1 and FS2 in
    Theorem~\ref{thm:B}) directly from the bare extensional Kan complex interface.
    A positive answer would unify the bare and coherent semantic regimes,
    showing that every extensional Kan complex is automatically coherent
    in the sense needed for associator and pentagon reasoning.

  \item \textbf{Alternative higher towers on $\Kinf$.}
    The present paper uses the canonical identity-type tower on the carrier of
    $\Kinf$.  It would be interesting to compare this choice with more geometric
    higher structures, for example towers built from path objects or from the
    Scott topology, and to ask whether the fixed-span separation persists or
    strengthens there.

  \item \textbf{Typed extensions.}
    Extend the recursive-completion comparison and its higher-dimensional
    consequences to typed $\lambda$-calculi, where the type structure could
    provide additional constraints on the higher-cell tower.  Simply-typed and
    polymorphic calculi are natural first targets.

  \item \textbf{Comparison questions with HoTT.}
    Investigate whether the low-dimensional core $\mathscr{G}_\lambda$ or its
    recursive completion admits a useful comparison map to the fundamental
    higher-groupoid constructions of homotopy type theory~\cite{HoTT}.

  \item \textbf{Explicit strictification.}
    Determine whether the low-dimensional core $\mathscr{G}_\lambda$ or its
    recursive completion admits a strictification to a strict $\omega$-category,
    and if so, what information about the $\lambda$-calculus is lost in the
    process.
\end{enumerate}

\subsection{Formal verification note}

All results proved in this paper have been formally verified in the Lean~4
theorem prover~\cite{Lean4}; the complete development is available at
\begin{center}
  \url{https://github.com/Arthur742Ramos/HigherLambdaModel}
\end{center}
\noindent From the repository root, \texttt{lake build} checks the full
development.  The file \path{docs/theorem_index.md} is the exhaustive
paper-to-Lean claim matrix, and the project sources contain no local uses of
\texttt{sorry}, \texttt{admit}, or \texttt{axiom}.
From the mathematical point of view, the formalization plays an evidential role:
it certifies that every case split, boundary compatibility check, and recursive
induction invoked above has been verified mechanically.  In particular, the
recursive inductions in the proofs of Theorems~\ref{thm:A} and~\ref{thm:D}
involve several dozen case-by-case verifications that are straightforward but
tedious to reproduce on paper; the formal proof provides independent assurance
that none were dropped.
From a logic-and-computation perspective, the formalization also serves as a
computational certificate: the Lean kernel's type-checking algorithm
independently verifies every logical step, providing a degree of assurance that
complements traditional peer review.
The formalization also covers all four main theorems,
the baseline results of Section~\ref{sec:baseline}, and the shared
low-dimensional core through dimension~6 together with the recursively
completed tower used in Theorem~\ref{thm:A}.

To keep the present paper focused on the mathematics, we record only one short
code fragment in the main text---the \texttt{ProjectionPair} interface from
Section~\ref{sec:Kinfty}.  For readers who want a direct bridge back to the
formal development, Appendix~\ref{app:lean-map} records representative Lean
entry points for the four main theorems.


\appendix

\section{Lean roadmap for the main theorems}\label{app:lean-map}

This appendix gives a compact map from the four headline results of the paper
to representative declarations in the Lean development, together with short
illustrative excerpts.  It is not an exhaustive index; the full paper-to-Lean
claim matrix is maintained in \path{docs/theorem_index.md}.  The excerpts below
are intentionally short and omit some argument lists with \texttt{...} when the
omitted data are not mathematically central to the discussion in the paper.

\subsection*{Theorem A: comparison with the recursive completion}

\noindent\textbf{Primary files.}
\begin{quote}\raggedright
\path{HigherLambdaModel/Lambda/Coherence.lean}\\
\path{HigherLambdaModel/Lambda/TruncationCore.lean}\\
\path{HigherLambdaModel/Lambda/Truncation.lean}
\end{quote}
\noindent\textbf{Representative declarations.}
\begin{quote}\raggedright
\path{AllDimensionalHigherConversionCoherence}\\
\path{lambdaOmegaConstructiveRealize}\\
\path{lambdaConstructiveHigherConversionCoherence}
\end{quote}

\begin{lstlisting}[caption={Representative Lean excerpt for Theorem A
  (from \texttt{Lambda/TruncationCore.lean}).}]
def lambdaConstructiveHigherConversionCoherence :
    HigherLambdaModel.Lambda.Coherence.AllDimensionalHigherConversionCoherence
      HigherLambdaModel.Lambda.NTerms.lambdaTower
      lambdaOmegaReflexiveTower where
  realize := lambdaOmegaConstructiveRealize
  source_comm := by
    intro n x
    exact lambdaOmegaConstructiveRealize_source_comm x
  target_comm := by
    intro n x
    exact lambdaOmegaConstructiveRealize_target_comm x
\end{lstlisting}

\subsection*{Theorem B: front-seed sufficiency}

\noindent\textbf{Primary file.}
\begin{quote}\raggedright
\path{HigherLambdaModel/Lambda/ExtensionalKanHigher.lean}
\end{quote}
\noindent\textbf{Representative declarations.}
\begin{quote}\raggedright
\path{FrontSeedCoherentExtensionalKanComplex}\\
\path{FrontSeedCoherentExtensionalKanComplex.reductionSeq_comp_associator_in_Theory3}\\
\path{FrontSeedCoherentExtensionalKanComplex.reductionSeq_pentagon_in_Theory3}\\
\path{FrontSeedCoherentExtensionalKanComplex.homotopy2_pentagon_source_bridge_in_Theory3}\\
\path{FrontSeedCoherentExtensionalKanComplex.homotopy2_pentagon_target_bridge_in_Theory3}\\
\path{FrontSeedCoherentExtensionalKanComplex.homotopy2_pentagon_shell_bridge_in_Theory3}
\end{quote}

The source bridge lands in the structural semantic source shell.  The target and
shell bridges land in the mixed target shell where the right-whiskered factor
remains equality-generated; the fully structural pentagon comparison itself is
recorded separately by
\path{FrontSeedCoherentExtensionalKanComplex.reductionSeq_pentagon_in_Theory3}.

\begin{lstlisting}[caption={Representative Lean excerpt for Theorem B
  (from \texttt{Lambda/ExtensionalKanHigher.lean}).}]
structure FrontSeedCoherentExtensionalKanComplex
    extends ExtensionalKanComplex where
  pentagonInnerRightFrontReflPath3 : ...
  wlwrPath3 : ...

noncomputable def
    FrontSeedCoherentExtensionalKanComplex.reductionSeq_comp_associator_in_Theory3
    (K : FrontSeedCoherentExtensionalKanComplex) ... :=
  reductionSeq_comp_associator_in_Theory3_ofPentagonInnerRightFrontRefl
    K.toExtensionalKanComplex
    K.pentagonInnerRightFrontReflPath3
    K.wlwrPath3 ...
\end{lstlisting}

\subsection*{Theorem C: fixed-span persistence in the canonical $\Kinf$ tower}

\noindent\textbf{Primary files.}
\begin{quote}\raggedright
\path{HigherLambdaModel/KInfinity/Examples.lean}\\
\path{HigherLambdaModel/KInfinity/Properties.lean}\\
\path{HigherLambdaModel/KInfinity/ContinuousSemantics.lean}\\
\path{HigherLambdaModel/Lambda/NTerms.lean}
\end{quote}
\noindent\textbf{Representative declarations.}
\begin{quote}\raggedright
\path{example42NTerm1WitnessTag}\\
\path{example42NTerm1Witness_interpretation}\\
\path{betaEtaPaper_beta1Witness_interpretation}\\
\path{betaEtaPaper_eta1Witness_interpretation}\\
\path{betaEtaPaper_nterm1Witness_interpretations_distinct}\\
\path{betaEtaPaper_nterm1Witness_interpretations_no_path}\\
\path{betaEtaPaper_nterm1Witness_interpretations_no_2cell}\\
\path{betaEtaPaper_nterm1Witness_interpretations_no_3cell}\\
\path{betaEtaPaper_nterm1Witness_interpretations_no_4cell}\\
\path{betaEtaPaper_nterm1Witness_interpretations_no_5cell}\\
\path{betaEtaPaper_nterm1Witness_interpretations_no_recursive_higher_cell}
\end{quote}

\begin{lstlisting}[caption={Representative Lean excerpt for Theorem C
  (from \texttt{KInfinity/Examples.lean}).}]
noncomputable def betaEtaPaper_beta1Witness_interpretation :
    Example42NTerm1WitnessInterpretation :=
  example42NTerm1Witness_interpretation betaEtaPaper_beta1Witness

noncomputable def betaEtaPaper_eta1Witness_interpretation :
    Example42NTerm1WitnessInterpretation :=
  example42NTerm1Witness_interpretation betaEtaPaper_eta1Witness

theorem betaEtaPaper_nterm1Witness_interpretations_no_recursive_higher_cell
    (n : Nat) :
    not (Nonempty {x : kInfinityTower.Cell (n + 5) // ... }) := ...
\end{lstlisting}

\subsection*{Theorem D: exact reflexive packaging of $\Kinf$}

\noindent\textbf{Primary file.}
\begin{quote}\raggedright
\path{HigherLambdaModel/KInfinity/Properties.lean}
\end{quote}
\noindent\textbf{Representative declarations.}
\begin{quote}\raggedright
\path{hInfinity}\\
\path{kInfinityContinuous}\\
\path{kInfinityReflexiveCHPO}\\
\path{hInfinity_kInfinityContinuous_apply}\\
\path{kInfinityContinuous_hInfinityContinuous_apply}\\
\path{remark_4_3}\\
\path{Proposition43Witness}\\
\path{applicationContinuous}
\end{quote}

\begin{lstlisting}[caption={Representative Lean excerpt for Theorem D
  (from \texttt{KInfinity/Properties.lean}).}]
noncomputable def kInfinityReflexiveCHPO : ReflexiveCHPO KInfinityCHPO where
  reify := hInfinityContinuous
  reflect := kInfinityContinuous
  retract := by
    ext x
    exact hInfinity_kInfinityContinuous_apply x

theorem kInfinityContinuous_hInfinityContinuous_apply
    (f : ContinuousMap KInfinityCHPO KInfinityCHPO) :
    kInfinityContinuous (hInfinityContinuous f) = f := by
  ext x
  exact kInfinityContinuous_hInfinityContinuous_restrict_apply f x

noncomputable def applicationContinuous :
    ContinuousMap (Product.chpo KInfinityCHPO KInfinityCHPO) KInfinityCHPO := ...
\end{lstlisting}


\end{document}